\documentclass[final,runningheads,envcountsect,envcountsame]{llncs}
\usepackage{ifthen}
\newboolean{showappendix}
\newboolean{cropmargins}

\setboolean{showappendix}{true}
\setboolean{cropmargins}{true}

\ifthenelse{\boolean{cropmargins}}{
  \usepackage[marginparwidth=2cm,nohead,nofoot,
  headsep = 1cm,
  footskip = 1cm,
  text={12.2cm,19.3cm},
  paperwidth=16.2cm,
  paperheight=23.3cm,
  marginparsep=1mm,
  marginparwidth=1.9cm,
  footskip=1cm,
  centering
  ]{geometry}
}{
  \usepackage[paperheight=235mm, paperwidth=155mm,textwidth=12.2cm,textheight=19.3cm]{geometry}
}

\newcommand{\M}{\mathcal{M}}
\newcommand{\N}{\mathcal{N}}
\newcommand{\Nat}{\mathbb{N}}
\newcommand{\Obs}{\mathcal{T}}
\newcommand{\Hyp}{\mathcal{H}}
\newcommand{\bigO}{\mathcal{O}}
\newcommand{\lsharp}{L^{\#}}
\newcommand{\lstar}{L^{\ast}}
\newcommand{\converges}{\ensuremath{\mathord{\downarrow}}}
\newcommand{\diverges}{\ensuremath{\mathord{\uparrow}}}

\newcommand{\code}[1]{\texttt{#1}}

\usepackage[utf8]{inputenc}
\usepackage{amsmath}
\usepackage{amssymb}
\usepackage{pbox}
\usepackage{tabularx}
\usepackage[T1]{fontenc}
\usepackage{color}
\usepackage{colortbl}
\usepackage{graphicx}
\usepackage{pifont}
\usepackage{ifdraft}
\usepackage{booktabs}
\usepackage{microtype}
\usepackage{epigraph}
\usepackage{subcaption}
\usepackage{enumitem}
\usepackage{float}
\usepackage{stmaryrd}
\usepackage{bbding}
\usepackage[outline]{contour}
\contourlength{1.2pt}

\usepackage{tikz}
\usetikzlibrary{patterns, intersections, positioning, shapes.geometric, calc,
	automata, arrows, decorations.pathreplacing,decorations.pathmorphing}
\usetikzlibrary{cd}

\newcommand{\treeNodeLabel}[1]{\contour{white}{#1}}

\tikzset{
  initial text={},
	treenode/.style = {align=center, inner sep=0pt, text centered},
  basis/.style = {
    pattern=north east lines,
    pattern color=magenta!80!black!80!white,
  },
  frontier/.style={
    pattern=crosshatch dots,
    pattern color=yellow!80!black,
  },
  q0class/.style={
    pattern=vertical lines,
    pattern color=red!60!white,
  },
  q1class/.style={
    pattern=north west lines,
    pattern color=blue!60!white,
  },
  q2class/.style={
    pattern=crosshatch,
    pattern color=green!50!black!60!white,
  },
}

\usepackage{xspace}
\usepackage{multirow}
\usepackage{float}
\usepackage{bold-extra}
\usepackage{algorithm,algpseudocode}
\usepackage{aliascnt}

\makeatletter
\RequirePackage[bookmarks,unicode,colorlinks=true]{hyperref}%
\def\@citecolor{blue}%
\def\@urlcolor{blue}%
\def\@linkcolor{blue}%
\def\orcidID#1{\smash{\protect\raisebox{-1.25pt}{\protect\href{http://orcid.org/#1}{\includegraphics{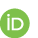}}}}}
\makeatother

\hypersetup{final,hidelinks}
\makeatletter
\hypersetup{
    final,
    hidelinks,
    urlcolor=black,
    pdftitle={A New Approach for Active Automata Learning Based on Apartness},
    pdfauthor={Frits Vaandrager, Bharat Garhewal, Jurriaan Rot, Thorsten Wißmann},
    pdfkeywords={},
}
\makeatother

\usepackage[capitalise]{cleveref}
\newcommand{\appendixref}[1]{Appendix~\ref{#1}}

\tikzset{
  do guard/.style={
    inner xsep=0pt,
  },
  guard line offset/.style={
    xshift=2mm,
  },
  bigtalloblong/.style={
    draw=black,
    minimum width=3pt,
    minimum height=1em,
    inner xsep=0pt,
  },
}
\newcommand{\connectDoGuards}[2]{%
  \draw[overlay,draw=black!40!white] ([guard line offset]#1) -- ([guard line offset]#2);
}

\algblockdefx{While}{EndWhile}[1]{
  \textbf{while}~\text{#1}
}{
  \textbf{end while}
}

\algblockdefx{DoIf}{EndDoIf}
[1]{\begin{tikzpicture}[remember picture,baseline=(do.base)]
    \node[do guard] (do) {\textbf{do}};
    \coordinate (last do) at (do.south west);
  \end{tikzpicture}
     #1 $\rightarrow$}
   {\begin{tikzpicture}[remember picture,baseline=(od.base)]
       \node[do guard] (od) {X};
       \connectDoGuards{opening do}{od.base west}
    \end{tikzpicture}\textbf{od}}
\algcblockdefx[DoIf]{DoIf}{ElsDoIf}{EndDoIf}
[1]{\begin{tikzpicture}[remember picture,baseline=(new do.base)]
    \node[do guard] (new do) {\phantom{\bfseries d}};
    \draw[draw=none] ([yshift=3pt]new do.north);
    \node[bigtalloblong,anchor=center] at ([guard line offset]new do.west) {};
    \connectDoGuards{last do}{new do.north west}
    \coordinate (last do) at (new do.south west);
    \end{tikzpicture}
  #1 $\rightarrow$}
{\begin{tikzpicture}[remember picture,baseline=(od.base)]
       \node[do guard] (od) {\textbf{end do}};
       \connectDoGuards{last do}{od.north west}
    \end{tikzpicture}%
  }

\newcommand{\StateIf}[1]{\State\textbf{if}~#1~\textbf{then:}}
\newcommand{\StateElse}{\State\textbf{else:}~}
\algloopdefx{IfLoop}[1]{\textbf{If} #1 \textbf{then}}

\usepackage{twshowkeys}

\usepackage[marginclue,footnote]{fixme}
\FXRegisterAuthor{tw}{atw}{TW}%
\FXRegisterAuthor{fv}{afv}{FV}%
\FXRegisterAuthor{jr}{ajr}{JR}%
\FXRegisterAuthor{bg}{abg}{BG}%

\newcommand{\id}{\ensuremath{\mathsf{id}}}
\newcommand{\access}{\ensuremath{\mathsf{access}}}
\newcommand{\partialto}{\ensuremath{\rightharpoonup}}
\newcommand{\fpair}[1]{\ensuremath{{\langle #1 \rangle}}}
\newcommand{\textqt}[1]{`#1'}
\newcommand{\apart}{\ensuremath{\mathrel{\#}}}
\newcommand{\defined}{\ensuremath{\mathord{\downarrow}}}
\newcommand{\semantics}[1]{\ensuremath{{\llbracket #1\rrbracket}}}
\newcommand{\stateCan}[2]{\ensuremath{\delta(#1,#2)\converges}}
\newcommand{\qedhere}{\qed}
\newcommand{\doiurl}[1]{\href{https://dx.doi.org/#1}{\texttt{#1}}}
\newcommand{\citet}[1]{\cite{#1}}
\newcommand{\takeout}[1]{\relax}
\newenvironment{proofappendix}[2][Proof of]{%
  \subsection*{#1~\autoref{#2}}%
    \addcontentsline{toc}{subsection}{#1~\autoref{#2}}
  }{%
}

\newtheorem{assumption}[theorem]{Assumption}

\newcommand{\defineApiFunction}[2][]{%
  \ifthenelse{\equal{#1}{}}{
    \expandafter\newcommand\csname #2\endcsname{\text{\upshape\textsc{%
          #2%
    }}\xspace}%
  }{%
    \expandafter\newcommand\csname #2\endcsname{\text{\upshape\textsc{%
          #1%
        }}\xspace}%
  }%
}
\defineApiFunction{CheckConsistency}
\defineApiFunction{BuildHypothesis}
\defineApiFunction[EquivQuery]{EquivalenceQuery}
\defineApiFunction[ProcCounterEx]{ProcessCounterexample}
\defineApiFunction{OutputQuery}
\defineApiFunction{Ads}

\usepackage{etoolbox}

\AtBeginEnvironment{definition}{}

\usepackage{rotating}
\usepackage{adjustbox} %

\begin{document}
\title{A New Approach for Active Automata Learning Based on
  Apartness\thanks{%
    Research supported by NWO TOP project 612.001.852 ``Grey-box learning of Interfaces for Refactoring Legacy Software (GIRLS)''.}}
\titlerunning{A New Approach for Active Automata Learning}
\author{Frits Vaandrager$^\text{\Envelope}$ \orcidID{0000-0003-3955-1910}
	 \and
  Bharat Garhewal \orcidID{0000-0003-4908-2863} \and \\
  Jurriaan Rot \and
  Thorsten Wi{\ss}mann \orcidID{0000-0001-8993-6486}
}
\authorrunning{F. Vaandrager et al.}
\institute{
  Institute for Computing and Information Sciences, \\
  Radboud University, Nijmegen, the Netherlands
}
\maketitle              %
\begin{abstract}
We present $L^{\#}$, a new and simple approach to active automata learning. Instead of focusing on equivalence of observations, like the $L^{\ast}$ algorithm and its descendants, $L^{\#}$ takes a different perspective: it tries to establish \emph{apartness}, a constructive form of inequality. $L^{\#}$ does not require auxiliary notions such as observation tables or discrimination trees, but operates directly on tree-shaped automata. $L^{\#}$  has the same asymptotic query and symbol complexities as the best existing learning algorithms, but we show that adaptive distinguishing sequences can be naturally integrated to boost the performance of $L^{\#}$ in practice. Experiments with a prototype implementation, written in Rust, suggest that $L^{\#}$ is competitive with existing algorithms.
\end{abstract}

\section{Introduction}

In 1987, Dana Angluin published a seminal paper \cite{Ang87}, in which she showed that the class of regular languages can be learned efficiently using queries.  
In Angluin's approach of a \emph{minimally adequate teacher (MAT)}, learning is viewed as a game in which a learner has to infer
a deterministic finite automaton (DFA) for an unknown regular language $L$ by asking queries to a teacher.
The learner may pose two types of queries:
	``Is the word $w$ in $L$?'' (\emph{membership queries}), and
	``Is the language recognized by DFA $H$ equal to $L$?'' (\emph{equivalence queries}).
	In case of a \emph{no} answer to an equivalence query, the teacher supplies a counterexample that distinguishes hypothesis $H$ from $L$.
The $L^{\ast}$ algorithm proposed by Angluin \cite{Ang87} is able to learn $L$ by asking a polynomial
number of membership and equivalence queries (polynomial in the size of the corresponding canonical DFA).

Angluin's approach triggered a lot of subsequent research on active automata learning and has
numerous applications in the area of software and hardware analysis, for instance for generating conformance test
suites of software components \cite{HMNSBI2001}, finding bugs in
implementations of security-critical protocols
\cite{FJV16,FiterauEtAl17,FH17}, learning interfaces of classes in
software libraries \cite{HowarISBJ12},
inferring interface protocols of legacy software components \cite{AslamCSB20},
 and checking that a legacy
component and a refactored implementation have the same behavior
\cite{SHV16}. We refer to \cite{Vaa17,HowarS2018} for surveys and
further references.

Since 1987, major improvements of the original $L^{\ast}$ algorithm have been proposed, for instance by \cite{RivestS89,RivestS93,KearnsV94,MalerP95,ShahbazG09,IrfanOG10,MertenHSM11,Petrenko0GHO14,ThesisFalk,Isberner2014,Frohme19}. 
Yet, all these improvements are variations of $L^{\ast}$ in the sense that they approximate the Nerode congruence by means of refinement.
Isberner \cite{Isberner15} shows that these \emph{descendants} of $L^{\ast}$ can be described in a single, general framework.\footnote{Except for the ZQ algorithm of \cite{Petrenko0GHO14}, which was developed independently, and the ADT algorithm of \cite{Frohme19}, that was developed later and uses adaptive distinguishing sequences which are not covered in Isberner's framework.}

Variations of $L^{\ast}$ have also been used as a basis for learning extensions of DFAs such as
Mealy machines \cite{Nie03},
I/O automata \cite{AV10},
non-deterministic automata~\cite{BolligHKL09},
alternating automata~\cite{AngluinEF15},
register automata \cite{CEGAR12,CasselHJS16},
nominal automata \cite{MoermanS0KS17},
symbolic automata \cite{MalerM17,ArgyrosD18},
weighted automata \cite{Bergadano,BalleM15,HeerdtKR020},
Mealy machines with timers \cite{VBE21},
visibly pushdown automata \cite{Isberner15},
and categorical generalisations of automata~\cite{UrbatS20,Heerdt20,BarloccoKR19,ColcombetPS21}.
It is fair to say that $L^{\ast}$-like algorithms completely dominate the research area of active automata learning.

In this paper we present $\lsharp$, a fresh approach to automata learning that differs from $L^{\ast}$
and its descendants. Instead of focusing on equivalence of observations, $\lsharp$ tries to establish \emph{apartness},
a constructive form of inequality \cite{troelstra_schwichtenberg_2000,GJapartness}.  
The notion of apartness is standard in constructive real analysis and goes back to Brouwer, 
with Heyting giving an axiomatic treatment in \cite{Heyting27}. 
This change in perspective has several key consequences, developed and presented in this paper:
\begin{itemize}[topsep=5pt]
\item 
$\lsharp$ does not maintain auxiliary data structures such as observation tables or discrimination trees, but operates directly on the observation tree. This tree is a partial Mealy machine itself, and is very close to an actual hypothesis that can be submitted to the teacher.
As a result, our algorithm is \emph{simple}. 
\item 
The asymptotic query complexity of $\lsharp$ is $\bigO(k n^2 + n \log m)$ and
the asymptotic symbol complexity\footnote{The symbol complexity is the number of input symbols required to learn an automaton. This is a relevant measure for practical learning scenarios, where the total time needed to learn a model is proportional to the number of input symbols. }
 is $\bigO(k m n^2 + n m \log m)$.
Here $k$ is the number of input symbols, $n$ is the number of states, and $m$ is the length of the longest counterexample. 
These are the \emph{same asymptotic complexities} as the best existing ($L^{\ast}$-like) learning algorithms \cite{RivestS89,RivestS93,ThesisFalk,Isberner2014,Isberner15,Frohme19}.
\item 
The use of observation trees as primary data structure makes it easy to \emph{integrate 
concepts from conformance testing to improve the performance} of $\lsharp$.
In particular, adaptive distinguishing sequences \cite{LYa94}, which we can compute directly from the observation tree, turn out to be an effective 
boost in practice, even if their use does not affect asymptotic complexities.
Through $\lsharp$ testing and learning become even more intertwined \cite{BergGJLRS05,AichernigMMTT16}.
\item
Experiments on benchmarks of \cite{NeiderSVK97}, with a \emph{prototype
  implementation} written in Rust, suggest that $\lsharp$ is competitive with
existing, highly optimized algorithms implemented in LearnLib \cite{RSBM09}.
\end{itemize}

\paragraph{Related work.}
Despite the different data structures, $\lsharp$ and $\lstar$~\cite{Ang87} still
have many similarities, since both store all the information gained from all queries
so far. Moreover, both maintain a set of those states that have been learned
with absolute certainty already.
A few other algorithms have been proposed that follow a different approach than $L^{\ast}$.
Meinke \cite{Meinke10,MeinkeNS11} developed a dual approach where, instead of starting with a maximally coarse approximating relation and refining it during learning, one starts with a maximally fine relation and coarsens it by merging equivalence classes.  Although Meinke reports superior performance in the application to learning-based testing, these algorithms have exponential worst-case query complexities.
Using ideas from \cite{RivestS93}, Groz et al.\ \cite{GrozBSO20} use a combination of homing sequences and characterization sets to develop an algorithm for active model learning that does not require the ability to reset the system.
Via an extensive experimental evaluation involving benchmarks from \cite{NeiderSVK97} they show that the performance of their algorithm is competitive with the  $L^{\ast}$ descendant of \cite{ShahbazG09}, but there can be huge differences in the performance of their algorithm for models that are similar in size and structure.
Several authors have explored the use of SAT and SMT solvers for obtaining learning algorithms, see for instance \cite{PetrenkoAGO17,SmetsersFV18}, but these approaches suffer from fundamental scalability problems.
In a recent paper, Soucha \& Bogdanov~\cite{SouchaB20} outline an active learning algorithm which also takes the observation tree as the primary data structure, and use results from conformance testing to speed up learning.
They report that an implementation of their approach outperforms standard learning algorithms like $L^{\ast}$, but they have no explicit apartness relation and associated theoretical framework. It is precisely this theoretical underpinning which allowed us to establish complexity and correctness results, and define efficient procedures for counterexample processing and computing adaptive distinguishing sequences.

In the present paper, we first define partial Mealy machines, observation trees,
and apartness (\autoref{sec:mealy}). Then, we present the full $\lsharp$
algorithm (\autoref{sec:learning}) and benchmark our prototype implementation
(\autoref{sec:bench}).
\ifthenelse{\boolean{showappendix}}{%
  The proofs of all theorems can be found in \appendixref{proofs} and complete
  benchmark results in \appendixref{completeBenchmarks}.
}{%
  The proofs of all theorems and complete
  benchmark results can be found in the appendix of the full
  version~\cite{VGRW22} of this paper.
}

\section{Partial Mealy Machines and Apartness}
\label{sec:mealy}
The $\lsharp$ algorithm learns a hidden (complete) Mealy machine, and its primary
data structure is a \emph{partial} Mealy machine. We first fix notation
for partial maps.

We write $f \colon X \partialto Y$ to denote that $f$ is a partial function from $X$
to $Y$ and write $f(x) \converges$ to mean that $f$ is defined on $x$, that is,
$\exists y \in Y \colon f(x)=y$, and conversely write $f(x)\diverges$ if $f$ is
undefined for $x$.
Often, we identify a partial function $f \colon X \rightharpoonup Y$ with the set $\{ (x,y) \in X \times Y \mid f(x)=y \}$. The composition of partial
maps $f\colon X\partialto Y$ and $g\colon Y\partialto Z$ is denoted by $g\circ
f\colon X\partialto Z$, and we have $(g\circ f)(x)\converges$ iff $f(x)\converges$
and $g(f(x)) \converges$.
There is a partial order on
$X\partialto Y$ defined by $f\sqsubseteq g$
for $f,g\colon X\partialto Y$ if for all $x\in X$, $f(x)\converges$ implies $g(x)\converges$ and
$f(x) = g(x)$.

Throughout this paper, we fix a finite set $I$ of \emph{inputs} and a set $O$ of \emph{outputs}.

\begin{definition}
	\label{Mealy}
	A \emph{Mealy machine} is
	a tuple $\M = (Q, q_0, \delta, \lambda)$, where
	\begin{itemize}[beginpenalty=99,topsep=1pt]
	\item
	$Q$ is a finite set of \emph{states}
  and
	$q_0 \in Q$ is the \emph{initial state},
	\item
	$\fpair{\lambda,\delta}\colon Q \times I \partialto O\times Q$ is a partial map
  whose components are
  an \emph{output function} $\lambda\colon Q \times I \partialto O$
  and a \emph{transition function} $\delta\colon Q \times I \partialto Q$
  (hence, $\delta(q,i)\converges \Leftrightarrow \lambda(q,i)
  \converges$, for $q \in Q$ and $i \in I$).
	\end{itemize}
	We use superscript $\M$ to disambiguate to which Mealy machine we refer, e.g.\  $Q^{\M}$, $q^{\M}_0$, $\delta^{\M}$ and $\lambda^{\M}$.
  We write $q\xrightarrow{i/o}q'$, for $q,q'\in Q$, $i\in I$, $o\in O$ to denote $\lambda(q,i) =
  o$ and $\delta(q,i) = q'$.
	We call $\M$ \emph{complete} if $\delta$ is total, i.e., $\delta(q,i)$ is defined for all states $q$ and inputs $i$.
	We generalize the transition and output functions to input words of length
  $n\in \Nat$ by composing $\fpair{\lambda,\delta}$ $n$ times with itself:
  we define maps $\fpair{\lambda_n, \delta_n} \colon Q \times I^{n} \rightarrow O^n \times Q$
  by $\fpair{\lambda_0, \delta_0} = \id_Q$ and 
  \twnote{}
  \[
    \fpair{\lambda_{n+1},\delta_{n+1}}\colon
    \begin{tikzcd}[column sep=20mm]
    Q\times I^{n+1}
    \arrow[harpoon]{r}[overlay]{\fpair{\lambda_n,\delta_n}\times \id_{I}}
    &O^n \times Q\times I
    \arrow[harpoon]{r}
          [pos=0.5,overlay]{\id_{O^{n}}\times \fpair{\lambda,\delta}}
    &
    O^{n+1} \times Q
    \end{tikzcd}
  \]
  Whenever it is clear from the context, we use
  $\lambda$ and $\delta$ also for words.
\end{definition}

\begin{definition}
  The semantics of a state $q$ is a map $\semantics{q}\colon I^*\partialto O^*$
  defined by $\semantics{q}(\sigma) = \lambda(q, \sigma)$.
States $q, q'$ in possibly different Mealy machines are \emph{equivalent},
written $q \approx q'$, if $\semantics{q} = \semantics{q'}$.
Mealy machines $\M$ and $\N$ are \emph{equivalent} if their respective initial
states are equivalent: $q_0^\M \approx q_0^\N$.
\end{definition}

In our learning setting, an \emph{undefined} value in the partial transition map
represents lack of knowledge. We consider maps between Mealy
machines that preserve existing transitions, but possibly extend the
knowledge of transitions:
\begin{definition}
	\label{def refinement}
  For Mealy machines $\M$ and $\N$, a \emph{functional simulation}
  $f\colon \M\to \N$ is a map $f\colon Q^{\M} \to Q^{\N}$ with
  \begin{center}
  $f(q^{\M}_0)= q^{\N}_0$
    \qquad and\qquad $q\xrightarrow{i/o}q'$ implies $f(q)\xrightarrow{i/o}f(q')$.
  \end{center}
\end{definition}
Intuitively, a functional simulation preserves transitions. In the
literature, a functional simulation is also called \emph{refinement
  mapping} \cite{AL91}.

\begin{lemma}
	\label{la:refinement}
  For a functional simulation $f\colon \M\to \N$ and
  $q\in Q^\M$, we have $\semantics{q}\sqsubseteq \semantics{f(q)}$.
\end{lemma}

For a given machine $\M$, an observation tree is simply a Mealy machine itself which
represents the inputs and outputs we have observed so far during learning. 
Using functional simulations, we define it formally as follows.

\begin{definition}[(Observation) Tree]
	A Mealy machine $\Obs$ is 
	a \emph{tree} if for each $q \in Q^{\Obs}$ there is a unique sequence $\sigma \in I^*$ s.t.\ $\delta^{\Obs}(q^{\Obs}_0, \sigma) = q$.
	We write $\mathsf{access}(q)$ for the sequence of inputs leading to $q$.	
 A tree $\Obs$ is an \emph{observation tree} for a Mealy machine $\M$ if
  there is a functional simulation $f\colon \Obs\to \M$.
\end{definition}

Figure~\ref{Fig:Obs} shows an observation tree for the Mealy machine displayed on the right. The functional simulation $f$ is indicated via coloring of the states.
\begin{figure}[h]
	\begin{center}
		\begin{tikzpicture}[->,>=stealth',shorten >=1pt,auto,node distance=1.5cm,main node/.style={circle,draw,font=\sffamily\large\bfseries},
      ]
    \def\yoffset{8mm}
		\node[initial,state,q0class] (0) {\treeNodeLabel{$t_0$}};
		\node[state,q0class] (1) [right of=0,yshift=\yoffset] {\treeNodeLabel{$t_1$}};
		\node[state,q1class] (2) [right of=0,yshift=-\yoffset] {\treeNodeLabel{$t_2$}};
		\node[state,q2class] (3) [right of=2] {\treeNodeLabel{$t_3$}};
		\node[state,q0class] (5) [right of=2,yshift=2*\yoffset] {\treeNodeLabel{$t_5$}};
		\node[state,q1class] (4) [right of=3] {\treeNodeLabel{$t_4$}};

		\node[initial,state,q0class] [right of=5,xshift=2cm](q0) {\treeNodeLabel{$q_0$}};
		\node[state,q1class] (q1) [right of=q0] {\treeNodeLabel{$q_1$}};
		\node[state,q2class] (q2) [below of=q1] {\treeNodeLabel{$q_2$}};
    
    \node[anchor=base] (f) at ($ (5.base) !.45! (q0.base) $) {\begin{tikzcd}
        {} \arrow{r}{f}
        &[8mm] {}
      \end{tikzcd}};
		
		\path[every node/.style={font=\sffamily\scriptsize}]
		(0) edge node[sloped,above] {a/A} (1)
		    edge node[sloped,below] {b/B} (2)
		(2) edge node[above] {b/B} (3)
		    edge node[sloped,above] {a/A} (5)
		(3) edge node {a/C} (4)
		(q0) edge [bend left, text width=0.5cm] node {b/B} (q1)
		edge [loop below] node {a/A} (q0)
		(q1) edge [bend left, text width=0.5cm] node {b/B} (q2)
		edge [bend left, text width=0.5cm] node {a/A} (q0)
		(q2) edge [bend left] node[pos=0.4] {a/C} (q1)
		edge [loop right] node {b/B} (q2);
		\end{tikzpicture}
		\caption{An observation tree (left) for a Mealy machine (right).}
		\label{Fig:Obs}
	\end{center}
\end{figure}
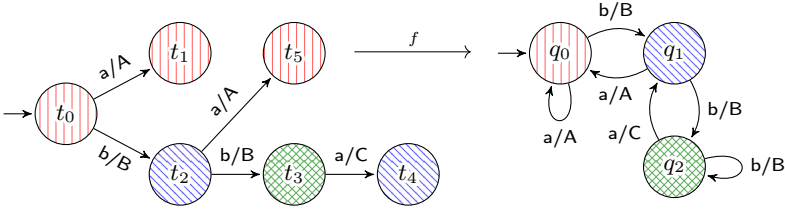

By performing output and equivalence queries, the learner can build an observation tree for the unknown Mealy machine $\M$ of the teacher. However, the learner does not know the functional simulation. Nevertheless, by analysis of the observation tree, the learner may infer that certain states in the tree cannot have the same color, that is, they cannot be mapped to same states of $\M$ by a functional simulation.
In this analysis, the concept of \emph{apartness}, a constructive form of inequality,  plays a crucial role \cite{troelstra_schwichtenberg_2000,GJapartness}.  
A similar concept has previously been studied in the context of automata learning under the name \emph{inequivalence constraints} in work on passive learning of DFAs, see for instance \cite{BiermannF72,FlorencioV14}.

\begin{definition}
	For a Mealy machine $\M$, we say that states $q,p\in Q^{\M}$ are \emph{apart}
  (written $q \apart p$) if there is some $\sigma\in I^*$ such that
  $\semantics{q}(\sigma)\converges$, $\semantics{p}(\sigma)\converges$,
  and $\semantics{q}(\sigma) \neq \semantics{p}(\sigma)$.
  We say that $\sigma$ is the \emph{witness} of $q\apart p$ and write $\sigma
  \vdash q\apart p$.
\end{definition}
Note that the apartness relation $\mathord{\apart}\subseteq Q\times Q$ is irreflexive and symmetric. 
A witness is also called \emph{separating sequence}~\cite{SmetsersMJ16}.
For the observation tree of Figure~\ref{Fig:Obs} we may derive the following apartness pairs and corresponding witnesses:
\[
  a \vdash t_0 \apart t_3 \qquad
  a \vdash t_2 \apart t_3 \qquad
  b \; a \vdash t_0 \apart t_2
\]
The apartness of states $q\apart p$ expresses that there is a conflict in their
semantics, and consequently, apart states can never be identified by a
functional simulation:
\begin{lemma}
	\label{la: apartness refinement}
  For a functional simulation $f\colon \Obs\to \M$,
  \[
    q\apart p\text{ in }\Obs
    \qquad\Longrightarrow
    \qquad
    f(q)\not\approx f(p)\text{ in }\M \qquad\text{for all }q, p\in Q^{\Obs}.
  \]
\end{lemma}
Thus, whenever states are apart in the observation tree $\Obs$, the learner
knows that these are distinct states in the hidden Mealy machine $\M$.

The apartness relation satisfies a weaker version of
\emph{co-transitivity}, stating that if $\sigma\vdash r\apart r'$ and $q$ has
the transitions for $\sigma$, then $q$ must be apart from at least one of $r$
and $r'$, or maybe even both:
\begin{lemma}[Weak co-transitivity]
	\label{la: weak co-transitivity}
  In every Mealy machine $\M$,
  \[
	\sigma \vdash r \apart r' ~\wedge~ \stateCan{q}{\sigma}  ~~\Longrightarrow~~
  r \apart q ~\vee~ r' \apart q
  \qquad\text{for all }r,r',q\in Q^{\M}, \sigma\in I^*.
\]
\end{lemma}
We use the weak co-transitivity property during learning. For instance in \autoref{Fig:Obs}, by posing the output query $aba$, consisting of the access sequence for $t_1$ concatenated with the witness $ba$ for $t_0 \apart t_2$, 
co-transitivity ensures that $t_0 \apart t_1$ or $t_2 \apart t_1$. By inspecting the outputs,
the learner may conclude that $t_2 \apart t_1$.

\section{Learning Algorithm}
\label{sec:learning}
The task solved by $\lsharp$ is to find a strategy for the learner in the following game:
\begin{definition}
  In the learning game between a learner and a teacher, the teacher has a
  complete Mealy machine $\M$ and answers the following queries from the
  learner:
  \begin{description}
  \item[$\OutputQuery(\sigma)$:] For $\sigma\in I^*$, the teacher replies with
    the corresponding output sequence $\lambda^{\M}(q_0^{\M},\sigma) \in O^*$.\footnote{In fact, later on we will assume that the teacher responds to slightly more
general output queries to enable the use of \emph{adaptive distinguishing sequences},
see \autoref{secAds}.
}
\item[$\EquivalenceQuery(\Hyp)$:]  For a complete Mealy
    machine $\Hyp$, the teacher replies $\code{yes}$ if $\Hyp\approx \M$ or
    $\code{no}$, providing some $\sigma\in I^*$ with
    $\lambda^{\M}(q_0^{\M},\sigma) \neq \lambda^{\Hyp}(q_0^{\Hyp},\sigma)$.
  \end{description}
\end{definition}

Our $\lsharp$ algorithm operates on an observation tree $\Obs = (Q, q_0, \delta, \lambda)$ for the unknown complete Mealy machine $\M$, where $\Obs$ contains the
results of all output and equivalence queries so far.
An observation tree is similar to the \emph{cache} which is commonly used in
implementations of $L^*$-based learning algorithms to store the answers to
previously asked queries, avoiding duplicates \cite{BalcazarDG97,MargariaRS05}.
But whereas for $L^*$-based learning algorithms the cache is an auxiliary data
structure and only used for efficiency reasons, it is a first-class citizen in $\lsharp$.
\begin{remark}
  The learner has no information about the teacher's hidden Mealy machine. In
  particular, whenever we write $\apart$, we always refer to the apartness
  relation \emph{on the observation tree} $\Obs$.
\end{remark}

The observation tree is structured in a very similar way as Dijkstra's shortest
path algorithm~\cite{Dijkstra59} structures a graph. Recall
that during the execution of Dijkstra's algorithm \textqt{the nodes are
  subdivided into three sets}~\cite{Dijkstra59}:
\begin{enumerate}
\item the nodes $S$ to which a shortest path from the initial node is known. $S$
  initially only contains the initial node and grows from there.
\item the nodes $F$ from which the next node to be added to $S$ will be selected.
\item the remaining nodes.
\end{enumerate}
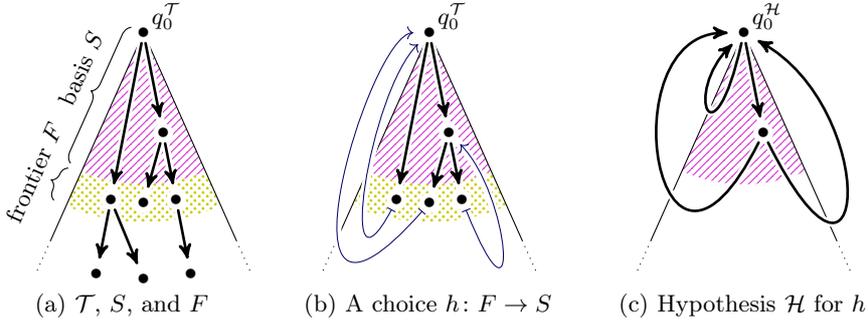
\begin{figure}[t]
  \begin{subfigure}{.33\textwidth}
    \centering
    \tikzset{
      obstree/.style={
        show feedback=false,
      },
    }
      \def\drawObsTransition#1{
    \begin{scope}[edge/.append style={
        every node/.style={opacity=0},
        }]
      \draw[edge,draw=white,line width=3pt,-] #1;
      \draw[edge,line width=2.5pt,-,draw=white,shorten >= -1pt] #1;
    \end{scope}
    \begin{pgfonlayer}{fg}
    \draw[edge,shorten >= 2pt] #1;
    \end{pgfonlayer}
  }%
  \pgfdeclarelayer{fg}
  \pgfsetlayers{main,fg}
    \begin{tikzpicture}[
      scale=1.0,
      brace scope/.style={},
      frontier transition scope/.style={},
      obstree/.append style={},
      show brace/.store in=\obstreeshowbrace,
      show brace=true,
      show feedback/.store in=\obstreeshowfeedback,
      show feedback=true,
      show frontier/.store in=\obstreeshowfrontier,
      show frontier=true,
      show frontier transition/.store in=\obstreeshowfrontiertrans,
      show frontier transition=true,
      show hyp transition/.store in=\obstreeshowhyptrans,
      show hyp transition=false,
      root name/.store in=\obstreeRootName,
      root name={$q_0^\Obs$},
      input word/.style={
        decorate,
        decoration={coil,aspect=0,amplitude=1pt},
        draw=white,
        double=black,
        double distance=1pt,
        line width=1pt,
      },
      edge/.style={
        ->,
        >=stealth',
        draw=black,
        line cap=round,
        line width=1pt,
        every node/.append style={
          pos=0.4,
          fill opacity=0,
          text opacity=0,
          font=\scriptsize,
          text=black,
          fill=white,
          inner sep=1pt,
          rounded corners=1pt,
          execute at begin node=\(,
          execute at end node=\),
        },
      },
      vertex/.style={
        shape=circle,
        inner sep=.5pt,
        fill=white,
      },
      iolabel/.style={
        fill=white,
        text=\footnotesize,
      },
      obstree,%
      ]
    \coordinate (root) at (0,0);
    \node[outer sep=0mm,anchor=south west]
          at (root) {\smash{\obstreeRootName}};
    \pgfmathsetmacro{\treeAngle}{24} %
    \pgfmathsetmacro{\sigmaAngle}{11} %
    \def\braceDistance{1pt} %
    \def\basisRadius{2cm} %
    \def\frontierRadius{2.5cm} %
    \def\obsTreeHeight{3cm} %
    \def\suffixlen{1cm} %
    \draw[basis,draw=none] (root) -- +(-90-\treeAngle:\basisRadius)
         arc (-90-\treeAngle:-90+\treeAngle:\basisRadius)
         -- cycle;
    \ifthenelse{\boolean{\obstreeshowfrontier}}{
    \draw[frontier,draw=none]
         ($ (root) +  (-90-\treeAngle:\basisRadius)    $)
         arc (-90-\treeAngle:-90+\treeAngle:\basisRadius)
         -- ++(-90+\treeAngle:\frontierRadius-\basisRadius)
         arc (-90+\treeAngle:-90-\treeAngle:\frontierRadius)
         -- cycle;
      }{}

    \draw (root) -- +(-90+\treeAngle:\obsTreeHeight)
          (root) -- +(-90-\treeAngle:\obsTreeHeight);
    \draw[dotted]
          ($ (root) + (-90+\treeAngle:\obsTreeHeight) $) -- +(-90+\treeAngle:5mm)
          ($ (root) + (-90-\treeAngle:\obsTreeHeight) $) -- +(-90-\treeAngle:5mm);
    \ifthenelse{\boolean{\obstreeshowbrace}}{
    \begin{scope}[
      mybrace/.style={
        decorate,
        decoration={
          brace,
          mirror,
          raise=5pt,
          amplitude=4pt,
        },
        every label/.append style={
          pos=0.5,sloped,
          text depth=32pt, %
        },
      },
      ]
    \draw [mybrace]
    (root) -- node[every label] {~~~~basis $S$}
    (-90-\treeAngle:\basisRadius-\braceDistance)
    ;
    \draw [mybrace]
    ($ (root) + (-90-\treeAngle:\basisRadius+\braceDistance) $)
    -- node[every label] {frontier $F$}
    (-90-\treeAngle:\frontierRadius)
    ;
    \end{scope}
    }{}

    \coordinate (q) at (-90+\sigmaAngle: .5*\basisRadius + .5*\frontierRadius);
    \coordinate (p) at ($ (q) + (-90+\sigmaAngle: \suffixlen)$);
    \coordinate (r) at (-90-\sigmaAngle: .5*\basisRadius + .5*\frontierRadius);
    \coordinate (m) at (-90: .5*\basisRadius + .5*\frontierRadius);
    \coordinate (low) at (-90: .5*\basisRadius + .5*\frontierRadius + \suffixlen);
    \coordinate (r plus sigma) at ($ (r) + (-90-\sigmaAngle: \suffixlen)$);

    \begin{pgfonlayer}{fg}
      \node[vertex] (root node) at (root) {$\bullet$};
      \node[vertex] (a node) at ($ (root) !.6! (q) $) {$\bullet$};
      \ifthenelse{\boolean{\obstreeshowfrontier}}{
        \node[vertex] (r node) at (r) {$\bullet$};
        \node[vertex] (m node) at (m) {$\bullet$};
        \node[vertex] (q node) at (q) {$\bullet$};
      }{}
    \end{pgfonlayer}
    \begin{scope}
      \drawObsTransition{(root node) to node{b/o} (a node)}
      \ifthenelse{\boolean{\obstreeshowfrontier}}{
        \drawObsTransition{(root node) to node[left]{a/o} (r node)}
        \drawObsTransition{(a node) to node[left]{a/p} (m node)}
        \drawObsTransition{(a node) to node[right]{b/o} (q node)}
      }{}
    \end{scope}
    \ifthenelse{\boolean{\obstreeshowfrontiertrans}}{
    \begin{scope}
      \begin{pgfonlayer}{fg}
        \node[overlay,vertex] (low  node) at (low) {$\bullet$};
        \node[overlay,vertex] (p node) at (p) {$\bullet$};
        \node[overlay,vertex] (r plus sigma node) at (r plus sigma) {$\bullet$};
      \end{pgfonlayer}
      \foreach \domain/\codomain in {%
        r node/r plus sigma node,
        q node/p node,
        r node/low node} {
        \drawObsTransition{(\domain) to (\codomain)}
      }
    \end{scope}
    }{}
    \ifthenelse{\boolean{\obstreeshowfeedback}}{
    \begin{scope}[
        edge/.append style={
          draw=blue!40!black,
          commutative diagrams/.cd,
          every diagram,
          every arrow,
          mapsto,
          overlay,
        },
      ]
      \drawObsTransition{
        (r node)
        .. controls +(-8mm,-15mm) and +(-8mm,-8mm)
        .. (root node)}
      \drawObsTransition{
        (m node)
        .. controls +(-22mm,-25mm) and +(-8mm,0mm)
        .. (root node)}
      \drawObsTransition{
        (q node)
        .. controls +(8mm,-20mm) and +(8mm,-8mm)
        .. (a node)}
    \end{scope}
    }{}
    \ifthenelse{\boolean{\obstreeshowhyptrans}}{
      \begin{scope}[overlay]
        \drawObsTransition{
          (root node)
          .. controls +(-4mm,-18mm) and +(-8mm,-8mm)
          .. (root node)}
        \drawObsTransition{
          (a node)
          .. controls +(-14mm,-28mm) and +(-18mm,-0mm)
          .. (root node)}
        \drawObsTransition{
          (a node)
          .. controls +(14mm,-28mm) and +(18mm,-10mm)
          .. (root node)}
      \end{scope}
    }{}
  \end{tikzpicture}
    \caption{$\Obs$, $S$, and $F$}
    \label{fig:obsTreeSF}
  \end{subfigure}%
  \begin{subfigure}{.33\textwidth}
    \centering
    \tikzset{
      obstree/.style={
        show brace=false,
        show feedback=true,
        show frontier transition=false,
      },
    }
      \def\drawObsTransition#1{
    \begin{scope}[edge/.append style={
        every node/.style={opacity=0},
        }]
      \draw[edge,draw=white,line width=3pt,-] #1;
      \draw[edge,line width=2.5pt,-,draw=white,shorten >= -1pt] #1;
    \end{scope}
    \begin{pgfonlayer}{fg}
    \draw[edge,shorten >= 2pt] #1;
    \end{pgfonlayer}
  }%
  \pgfdeclarelayer{fg}
  \pgfsetlayers{main,fg}
    \begin{tikzpicture}[
      scale=1.0,
      brace scope/.style={},
      frontier transition scope/.style={},
      obstree/.append style={},
      show brace/.store in=\obstreeshowbrace,
      show brace=true,
      show feedback/.store in=\obstreeshowfeedback,
      show feedback=true,
      show frontier/.store in=\obstreeshowfrontier,
      show frontier=true,
      show frontier transition/.store in=\obstreeshowfrontiertrans,
      show frontier transition=true,
      show hyp transition/.store in=\obstreeshowhyptrans,
      show hyp transition=false,
      root name/.store in=\obstreeRootName,
      root name={$q_0^\Obs$},
      input word/.style={
        decorate,
        decoration={coil,aspect=0,amplitude=1pt},
        draw=white,
        double=black,
        double distance=1pt,
        line width=1pt,
      },
      edge/.style={
        ->,
        >=stealth',
        draw=black,
        line cap=round,
        line width=1pt,
        every node/.append style={
          pos=0.4,
          fill opacity=0,
          text opacity=0,
          font=\scriptsize,
          text=black,
          fill=white,
          inner sep=1pt,
          rounded corners=1pt,
          execute at begin node=\(,
          execute at end node=\),
        },
      },
      vertex/.style={
        shape=circle,
        inner sep=.5pt,
        fill=white,
      },
      iolabel/.style={
        fill=white,
        text=\footnotesize,
      },
      obstree,%
      ]
    \coordinate (root) at (0,0);
    \node[outer sep=0mm,anchor=south west]
          at (root) {\smash{\obstreeRootName}};
    \pgfmathsetmacro{\treeAngle}{24} %
    \pgfmathsetmacro{\sigmaAngle}{11} %
    \def\braceDistance{1pt} %
    \def\basisRadius{2cm} %
    \def\frontierRadius{2.5cm} %
    \def\obsTreeHeight{3cm} %
    \def\suffixlen{1cm} %
    \draw[basis,draw=none] (root) -- +(-90-\treeAngle:\basisRadius)
         arc (-90-\treeAngle:-90+\treeAngle:\basisRadius)
         -- cycle;
    \ifthenelse{\boolean{\obstreeshowfrontier}}{
    \draw[frontier,draw=none]
         ($ (root) +  (-90-\treeAngle:\basisRadius)    $)
         arc (-90-\treeAngle:-90+\treeAngle:\basisRadius)
         -- ++(-90+\treeAngle:\frontierRadius-\basisRadius)
         arc (-90+\treeAngle:-90-\treeAngle:\frontierRadius)
         -- cycle;
      }{}

    \draw (root) -- +(-90+\treeAngle:\obsTreeHeight)
          (root) -- +(-90-\treeAngle:\obsTreeHeight);
    \draw[dotted]
          ($ (root) + (-90+\treeAngle:\obsTreeHeight) $) -- +(-90+\treeAngle:5mm)
          ($ (root) + (-90-\treeAngle:\obsTreeHeight) $) -- +(-90-\treeAngle:5mm);
    \ifthenelse{\boolean{\obstreeshowbrace}}{
    \begin{scope}[
      mybrace/.style={
        decorate,
        decoration={
          brace,
          mirror,
          raise=5pt,
          amplitude=4pt,
        },
        every label/.append style={
          pos=0.5,sloped,
          text depth=32pt, %
        },
      },
      ]
    \draw [mybrace]
    (root) -- node[every label] {~~~~basis $S$}
    (-90-\treeAngle:\basisRadius-\braceDistance)
    ;
    \draw [mybrace]
    ($ (root) + (-90-\treeAngle:\basisRadius+\braceDistance) $)
    -- node[every label] {frontier $F$}
    (-90-\treeAngle:\frontierRadius)
    ;
    \end{scope}
    }{}

    \coordinate (q) at (-90+\sigmaAngle: .5*\basisRadius + .5*\frontierRadius);
    \coordinate (p) at ($ (q) + (-90+\sigmaAngle: \suffixlen)$);
    \coordinate (r) at (-90-\sigmaAngle: .5*\basisRadius + .5*\frontierRadius);
    \coordinate (m) at (-90: .5*\basisRadius + .5*\frontierRadius);
    \coordinate (low) at (-90: .5*\basisRadius + .5*\frontierRadius + \suffixlen);
    \coordinate (r plus sigma) at ($ (r) + (-90-\sigmaAngle: \suffixlen)$);

    \begin{pgfonlayer}{fg}
      \node[vertex] (root node) at (root) {$\bullet$};
      \node[vertex] (a node) at ($ (root) !.6! (q) $) {$\bullet$};
      \ifthenelse{\boolean{\obstreeshowfrontier}}{
        \node[vertex] (r node) at (r) {$\bullet$};
        \node[vertex] (m node) at (m) {$\bullet$};
        \node[vertex] (q node) at (q) {$\bullet$};
      }{}
    \end{pgfonlayer}
    \begin{scope}
      \drawObsTransition{(root node) to node{b/o} (a node)}
      \ifthenelse{\boolean{\obstreeshowfrontier}}{
        \drawObsTransition{(root node) to node[left]{a/o} (r node)}
        \drawObsTransition{(a node) to node[left]{a/p} (m node)}
        \drawObsTransition{(a node) to node[right]{b/o} (q node)}
      }{}
    \end{scope}
    \ifthenelse{\boolean{\obstreeshowfrontiertrans}}{
    \begin{scope}
      \begin{pgfonlayer}{fg}
        \node[overlay,vertex] (low  node) at (low) {$\bullet$};
        \node[overlay,vertex] (p node) at (p) {$\bullet$};
        \node[overlay,vertex] (r plus sigma node) at (r plus sigma) {$\bullet$};
      \end{pgfonlayer}
      \foreach \domain/\codomain in {%
        r node/r plus sigma node,
        q node/p node,
        r node/low node} {
        \drawObsTransition{(\domain) to (\codomain)}
      }
    \end{scope}
    }{}
    \ifthenelse{\boolean{\obstreeshowfeedback}}{
    \begin{scope}[
        edge/.append style={
          draw=blue!40!black,
          commutative diagrams/.cd,
          every diagram,
          every arrow,
          mapsto,
          overlay,
        },
      ]
      \drawObsTransition{
        (r node)
        .. controls +(-8mm,-15mm) and +(-8mm,-8mm)
        .. (root node)}
      \drawObsTransition{
        (m node)
        .. controls +(-22mm,-25mm) and +(-8mm,0mm)
        .. (root node)}
      \drawObsTransition{
        (q node)
        .. controls +(8mm,-20mm) and +(8mm,-8mm)
        .. (a node)}
    \end{scope}
    }{}
    \ifthenelse{\boolean{\obstreeshowhyptrans}}{
      \begin{scope}[overlay]
        \drawObsTransition{
          (root node)
          .. controls +(-4mm,-18mm) and +(-8mm,-8mm)
          .. (root node)}
        \drawObsTransition{
          (a node)
          .. controls +(-14mm,-28mm) and +(-18mm,-0mm)
          .. (root node)}
        \drawObsTransition{
          (a node)
          .. controls +(14mm,-28mm) and +(18mm,-10mm)
          .. (root node)}
      \end{scope}
    }{}
  \end{tikzpicture}
    \caption{A choice $h\colon F\to S$}
    \label{fig:feedback}
  \end{subfigure}
  \begin{subfigure}{.33\textwidth}
    \centering
    \tikzset{
      obstree/.style={
        show brace=false,
        show feedback=false,
        show frontier=false,
        show frontier transition=false,
        show hyp transition=true,
        root name=$q_0^{\Hyp}$,
      },
    }
      \def\drawObsTransition#1{
    \begin{scope}[edge/.append style={
        every node/.style={opacity=0},
        }]
      \draw[edge,draw=white,line width=3pt,-] #1;
      \draw[edge,line width=2.5pt,-,draw=white,shorten >= -1pt] #1;
    \end{scope}
    \begin{pgfonlayer}{fg}
    \draw[edge,shorten >= 2pt] #1;
    \end{pgfonlayer}
  }%
  \pgfdeclarelayer{fg}
  \pgfsetlayers{main,fg}
    \begin{tikzpicture}[
      scale=1.0,
      brace scope/.style={},
      frontier transition scope/.style={},
      obstree/.append style={},
      show brace/.store in=\obstreeshowbrace,
      show brace=true,
      show feedback/.store in=\obstreeshowfeedback,
      show feedback=true,
      show frontier/.store in=\obstreeshowfrontier,
      show frontier=true,
      show frontier transition/.store in=\obstreeshowfrontiertrans,
      show frontier transition=true,
      show hyp transition/.store in=\obstreeshowhyptrans,
      show hyp transition=false,
      root name/.store in=\obstreeRootName,
      root name={$q_0^\Obs$},
      input word/.style={
        decorate,
        decoration={coil,aspect=0,amplitude=1pt},
        draw=white,
        double=black,
        double distance=1pt,
        line width=1pt,
      },
      edge/.style={
        ->,
        >=stealth',
        draw=black,
        line cap=round,
        line width=1pt,
        every node/.append style={
          pos=0.4,
          fill opacity=0,
          text opacity=0,
          font=\scriptsize,
          text=black,
          fill=white,
          inner sep=1pt,
          rounded corners=1pt,
          execute at begin node=\(,
          execute at end node=\),
        },
      },
      vertex/.style={
        shape=circle,
        inner sep=.5pt,
        fill=white,
      },
      iolabel/.style={
        fill=white,
        text=\footnotesize,
      },
      obstree,%
      ]
    \coordinate (root) at (0,0);
    \node[outer sep=0mm,anchor=south west]
          at (root) {\smash{\obstreeRootName}};
    \pgfmathsetmacro{\treeAngle}{24} %
    \pgfmathsetmacro{\sigmaAngle}{11} %
    \def\braceDistance{1pt} %
    \def\basisRadius{2cm} %
    \def\frontierRadius{2.5cm} %
    \def\obsTreeHeight{3cm} %
    \def\suffixlen{1cm} %
    \draw[basis,draw=none] (root) -- +(-90-\treeAngle:\basisRadius)
         arc (-90-\treeAngle:-90+\treeAngle:\basisRadius)
         -- cycle;
    \ifthenelse{\boolean{\obstreeshowfrontier}}{
    \draw[frontier,draw=none]
         ($ (root) +  (-90-\treeAngle:\basisRadius)    $)
         arc (-90-\treeAngle:-90+\treeAngle:\basisRadius)
         -- ++(-90+\treeAngle:\frontierRadius-\basisRadius)
         arc (-90+\treeAngle:-90-\treeAngle:\frontierRadius)
         -- cycle;
      }{}

    \draw (root) -- +(-90+\treeAngle:\obsTreeHeight)
          (root) -- +(-90-\treeAngle:\obsTreeHeight);
    \draw[dotted]
          ($ (root) + (-90+\treeAngle:\obsTreeHeight) $) -- +(-90+\treeAngle:5mm)
          ($ (root) + (-90-\treeAngle:\obsTreeHeight) $) -- +(-90-\treeAngle:5mm);
    \ifthenelse{\boolean{\obstreeshowbrace}}{
    \begin{scope}[
      mybrace/.style={
        decorate,
        decoration={
          brace,
          mirror,
          raise=5pt,
          amplitude=4pt,
        },
        every label/.append style={
          pos=0.5,sloped,
          text depth=32pt, %
        },
      },
      ]
    \draw [mybrace]
    (root) -- node[every label] {~~~~basis $S$}
    (-90-\treeAngle:\basisRadius-\braceDistance)
    ;
    \draw [mybrace]
    ($ (root) + (-90-\treeAngle:\basisRadius+\braceDistance) $)
    -- node[every label] {frontier $F$}
    (-90-\treeAngle:\frontierRadius)
    ;
    \end{scope}
    }{}

    \coordinate (q) at (-90+\sigmaAngle: .5*\basisRadius + .5*\frontierRadius);
    \coordinate (p) at ($ (q) + (-90+\sigmaAngle: \suffixlen)$);
    \coordinate (r) at (-90-\sigmaAngle: .5*\basisRadius + .5*\frontierRadius);
    \coordinate (m) at (-90: .5*\basisRadius + .5*\frontierRadius);
    \coordinate (low) at (-90: .5*\basisRadius + .5*\frontierRadius + \suffixlen);
    \coordinate (r plus sigma) at ($ (r) + (-90-\sigmaAngle: \suffixlen)$);

    \begin{pgfonlayer}{fg}
      \node[vertex] (root node) at (root) {$\bullet$};
      \node[vertex] (a node) at ($ (root) !.6! (q) $) {$\bullet$};
      \ifthenelse{\boolean{\obstreeshowfrontier}}{
        \node[vertex] (r node) at (r) {$\bullet$};
        \node[vertex] (m node) at (m) {$\bullet$};
        \node[vertex] (q node) at (q) {$\bullet$};
      }{}
    \end{pgfonlayer}
    \begin{scope}
      \drawObsTransition{(root node) to node{b/o} (a node)}
      \ifthenelse{\boolean{\obstreeshowfrontier}}{
        \drawObsTransition{(root node) to node[left]{a/o} (r node)}
        \drawObsTransition{(a node) to node[left]{a/p} (m node)}
        \drawObsTransition{(a node) to node[right]{b/o} (q node)}
      }{}
    \end{scope}
    \ifthenelse{\boolean{\obstreeshowfrontiertrans}}{
    \begin{scope}
      \begin{pgfonlayer}{fg}
        \node[overlay,vertex] (low  node) at (low) {$\bullet$};
        \node[overlay,vertex] (p node) at (p) {$\bullet$};
        \node[overlay,vertex] (r plus sigma node) at (r plus sigma) {$\bullet$};
      \end{pgfonlayer}
      \foreach \domain/\codomain in {%
        r node/r plus sigma node,
        q node/p node,
        r node/low node} {
        \drawObsTransition{(\domain) to (\codomain)}
      }
    \end{scope}
    }{}
    \ifthenelse{\boolean{\obstreeshowfeedback}}{
    \begin{scope}[
        edge/.append style={
          draw=blue!40!black,
          commutative diagrams/.cd,
          every diagram,
          every arrow,
          mapsto,
          overlay,
        },
      ]
      \drawObsTransition{
        (r node)
        .. controls +(-8mm,-15mm) and +(-8mm,-8mm)
        .. (root node)}
      \drawObsTransition{
        (m node)
        .. controls +(-22mm,-25mm) and +(-8mm,0mm)
        .. (root node)}
      \drawObsTransition{
        (q node)
        .. controls +(8mm,-20mm) and +(8mm,-8mm)
        .. (a node)}
    \end{scope}
    }{}
    \ifthenelse{\boolean{\obstreeshowhyptrans}}{
      \begin{scope}[overlay]
        \drawObsTransition{
          (root node)
          .. controls +(-4mm,-18mm) and +(-8mm,-8mm)
          .. (root node)}
        \drawObsTransition{
          (a node)
          .. controls +(-14mm,-28mm) and +(-18mm,-0mm)
          .. (root node)}
        \drawObsTransition{
          (a node)
          .. controls +(14mm,-28mm) and +(18mm,-10mm)
          .. (root node)}
      \end{scope}
    }{}
  \end{tikzpicture}
    \caption{Hypothesis $\Hyp$ for $h$}
    \label{fig:hyp}
  \end{subfigure}
    \caption{From the observation tree to the hypothesis ($|I| = 2$)}
\end{figure}
This scheme adapts to the observation tree as follows and is visualized in
\autoref{fig:obsTreeSF}.
\begin{enumerate}
\item The states $S\subseteq Q^\Obs$, which already have been fully identified,
  i.e.~the learner found out that these must represent distinct states in the
  teacher's hidden Mealy machine. We call $S$ the \emph{basis}. Initially, $S :=
  \{q_0^\Obs\}$, and throughout the execution $S$ forms a subtree of $\Obs$ and
  all states in $S$ are pairwise apart: $\forall p,q\in S, p\neq q\colon p\apart
  q$.
\item the \emph{frontier} $F\subseteq Q^\Obs$, from
  which the next node to be added to $S$ is chosen. Throughout the execution, $F$ is the
  set of immediate non-basis successors of basis states:
  $
    F ~:=~ \{ q' \in Q \setminus S \mid \exists q \in S, i \in I : q' = \delta(q, i) \}.
  $
\item the remaining states $Q\setminus (S\cup F)$.
\end{enumerate}

Initially, $\Obs$ consists
of only an initial state $q_0^{\Obs}$ with no transitions. For every
$\OutputQuery(\sigma)$ during the execution, the input $\sigma\in I^*$ and the corresponding response
of type $O^*$ is added automatically to the observation tree $\Obs$, and similarly
every negative response to a $\EquivalenceQuery$ leads to new states and
transitions in the observation tree.
With every extension $\Obs'$ of the
observation tree $\Obs$, the apartness relation can only grow:
whenever $p\apart q$ in $\Obs$, then still $p\apart q$ in $\Obs'$.
Thus, along the learning game, $\Obs$ and $\apart$ grow steadily:
\begin{assumption}
  We implicitly require that via output and equivalence queries, the observation
  tree $\Obs$ and the basis $S$ are gradually extended, with the frontier $F$
  automatically moving along while $S$ grows.
\end{assumption}

\subsection{Hypothesis construction}
At almost any point during the learning game, the learner can come up with a
hypothesis $\Hyp$ based on the knowledge in the observation tree $\Obs$. Since
the basis $S$ contains the states already discovered, the set of states
of such a hypothesis is simply set to $Q^{\Hyp} := S$, and it contains every
transition between basis states (in $\Obs$). The hypothesis must also reflect
the transitions in $\Obs$ that leave the basis $S$, i.e.~the transitions to the
frontier. Those are resolved by finding for every frontier state a base state,
for which the learner conjectures that they are equivalent states in the hidden
Mealy machine. This choice boils down to a map $h\colon F\to S$
($\mapsto$ in \autoref{fig:feedback}). Then, a transition $q\xrightarrow{i/o} p$ in $\Obs$
with $q\in S$, $p\in F$ leads to a transition $q\xrightarrow{i/o} h(p)$ in
$\Hyp$ (\autoref{fig:hyp}). These ideas are formally defined as follows.
\begin{definition}
	Let $\Obs$ be an observation tree with basis $S$ and frontier $F$.
  \begin{enumerate}
  \item A Mealy machine $\Hyp$ \emph{contains the basis}
    if $Q^{\Hyp} = S$ and $\delta^{\Hyp}(q_0^{\Hyp},\access(q)) = q$ for all $q\in S$.
  \item A \emph{hypothesis} is a complete Mealy machine $\Hyp$ containing the basis such that
      $q\xrightarrow{i/o'} p'$ in $\Hyp$ ($q\in S$)
      and $q\xrightarrow{i/o} p$ in $\Obs$ imply
      $o=o'$ and $\neg (p\apart p')$ (in $\Obs$).

  \item A hypothesis $\Hyp$ is \emph{consistent} if there is a functional
  simulation $f\colon \Obs\to \Hyp$.
  \item For a Mealy machine $\Hyp$ containing the basis, an input sequence $\sigma
  \in I^*$ is said to \emph{lead to a conflict} if
  $\delta^{\Obs}(q_0^\Obs,\sigma) \apart \delta^{\Hyp}(q_0^{\Hyp},\sigma)$ (in $\Obs$).
  \end{enumerate}
\end{definition}
Intuitively, the first three notions describe how confident we are in the correctness
of the \textqt{back loops} in $\Hyp$ obtained from a choice $h\colon F\to S$.
Notion 1 does not provide any warranty, notion 2 asserts that $\neg (q\apart
h(q))$ for all $q\in F$, and notion 3 (by definition) means that $\Obs$ is an
observation tree for $\Hyp$, that is, all observations so far are consistent
with the hypothesis $\Hyp$. The learner can verify the consistency of a hypothesis
without querying the teacher (algorithm is in \autoref{secConsistency} below). The
existence and uniqueness of a hypothesis are related to criteria on $\Obs$:

\begin{definition}
	In an observation tree $\Obs$,
	a state in $F$ is 1.~\emph{isolated} if it is apart from all states in $S$
  and 2.~is \emph{identified} if it is apart from all states in $S$ except one. 
  3.~The basis $S$ is \emph{complete} if each state in $S$ has a transition for each input in $I$.
\end{definition}

\begin{lemma} \label{hypExistence}
For an observation tree $\Obs$,
	if $F$ has no isolated states then there exists a hypothesis $\Hyp$ for $\Obs$.
	If $S$ is complete and all states in $F$ are identified then the hypothesis is unique.
\end{lemma}

With a growing observation tree $\Obs$, the hidden Mealy machine is found as
soon as the basis is big enough:

\begin{theorem} \label{hypSize}
	Suppose $\Obs$ is an observation tree for a (hidden) Mealy machine $\M$
	such that $S$ is complete, all states in $F$ are identified, and $|S|$ is
  the number of equivalence classes of $\approx^{\M}$. Then $\Hyp\approx \M$ for the unique hypothesis $\Hyp$.
\end{theorem}

The theorem itself is not necessary for the correctness of $\lsharp$, but
guarantees feasibility of learning.

\subsection{Main loop of the algorithm}
The $\lsharp$ algorithm is listed in \autoref{alg:lsharp} in pseudocode.\fvnote{}
The code uses Dijkstra's guarded command notation~\cite{Dij75}, which means that
the following rules are applied non-deterministically until none of them can be
applied anymore:
\begin{algorithm}[t]
	\begin{algorithmic}
		\Procedure{LSharp}{}
		\DoIf{$q$ isolated, for some $q \in F$}
      \Comment{Rule (R1)}
      \State $S \gets S \cup \{ q \}$
		\ElsDoIf{$\delta^{\Obs}(q,i)\uparrow$, for some $q \in S, i \in I$}
      \Comment{Rule (R2)}
      \State $\OutputQuery(\mathsf{access}(q) \; i)$
		\ElsDoIf{$\neg(q\apart r)$, $\neg(q\apart r')$, for some $q \in F$, $r,r'\in
      S$, $r\neq r'$}
      \Comment{Rule (R3)}
      \State $\sigma \gets \text{witness of \(r\apart r'\)}$
      \State $\OutputQuery(\mathsf{access}(q) \; \sigma)$
		\ElsDoIf{$F$ has no isolated states and basis $S$ is complete}
      \Comment{Rule (R4)}
		  \State $\Hyp \gets \BuildHypothesis$ 
		  \State $(b, \sigma) \gets \CheckConsistency(\Hyp)$		
		  \If{$b = \code{yes}$}
        \State $(b, \rho) \gets \EquivalenceQuery(\Hyp)$
        \StateIf{$b = \code{yes}$} \Return $\Hyp$
        \StateElse
          $\sigma \gets$ shortest prefix of $\rho$ such that
          $\delta^{\Hyp}(q_0^\Hyp, \sigma) \apart \delta^{\Obs}(q_0^\Obs, \sigma)$
        (in $\Obs$)
		  \EndIf
		  \State $\ProcessCounterexample(\Hyp, \sigma)$
    \EndDoIf
		\EndProcedure
	\end{algorithmic}
	\caption{Overall $L^{\#}$ algorithm}
	\label{alg:lsharp}
\end{algorithm}
\begin{description}
\item[(R1)]
If $F$ contains an isolated state, then this means that we have discovered a new
state not yet present in $S$, hence we move it from $F$ to $S$.

\item[(R2)] When a state $q \in S$ has no outgoing $i$-transition, for some $i \in I$,
  the output query for $\mathsf{access}(q) \; i$ will add the
  generated $i$ successor, implicitly extending the frontier $F$.

\item[(R3)] When $q\in F$ is a state in the frontier that is not yet identified,
  then there are at least two states in $S$ that are not apart from $q$. In this
  case, the algorithm picks a witness $\sigma\in I^*$ for $r\apart r'$. After the
  $\OutputQuery(\access(q)\,\sigma)$, the observation tree is extended and thus
  $q$ will be apart from at least $r$ or $r'$ by weak co-transitivity
  (\autoref{la: weak co-transitivity}).

\item[(R4)] When $F$ has no isolated states and $S$ is complete,
  \BuildHypothesis picks a hypothesis $\Hyp$ (at least one exists
  \autoref{hypExistence}). If $\Hyp$ is not consistent with observation tree
  $\Obs$ we get a conflict $\sigma$ for free. Otherwise, we pose an equivalence
  query for $\Hyp$. If the hypothesis is correct, $\lsharp$ terminates, and
  otherwise we obtain a counterexample $\rho$. The counterexample decomposes
  into two words $\sigma \eta$, where $\sigma$ leads to a conflict and $\eta$
  witnesses it. The conflict $\sigma$ means that one of the frontier states was
  merged with an apart basis state in $\Hyp$, causing a wrong transition in
  $\Hyp$. Since $\sigma$ can be very long, the task of
  $\ProcessCounterexample(\sigma)$ is to shorten $\sigma$ until we know which
  frontier state caused the conflict. So after $\ProcessCounterexample$, $\Hyp$
  is not a hypothesis for the updated $\Obs$ anymore.
\end{description}

We will show the correctness of $\lsharp$ in a top-down approach discussing the
subroutines later and only assuming now that:
\begin{enumerate}
\item $\BuildHypothesis$ picks one of the possible hypotheses (\autoref{hypExistence})
\item $\CheckConsistency(\Hyp)$ tells if there is a functional
  simulation $\Obs\to \Hyp$, and if not, provides $\sigma\in I^*$ leading to a
  conflict (\autoref{algConsistencyCorrect} below).
\item If $\Hyp$ contains the basis and $\sigma$ leads to a conflict, then
  $\ProcessCounterexample(\Hyp, \sigma)$, extends $\Obs$ such that $\Hyp$ is not
  a hypothesis anymore (\autoref{algProcCountCorrect} below).
\end{enumerate}

Whenever the algorithm terminates, the learner has found the correct model. 
Therefore, correctness amounts to showing termination.
The rough idea is that each rule will let $S$, $F$, or $\apart$ restricted to
$S\times F$ grow, and each of these sets are bounded by the hidden
Mealy machine $\M$. We define the norm $N(\Obs)$ by
\begin{equation}
\frac{ | S | \cdot (|S|+1)}{2} ~+~ | \{ (q, i) \in S \times I \mid \delta^\Obs(q,i) \converges \} | ~+~  | \{ (q, q') \in S \times F \mid q \apart q' \} |
  \label{norm}
\end{equation}
The first summand increases whenever a state is moved from $F$ to $S$ (R1);
it is quadratic in $|S|$ because (R1) reduces the third summand. The second
summand records the progress achieved by extending the frontier (R2). The third
summand counts how much the states in the frontier are identified (R3). Rule (R4)
extends the apartness relation, leading to an increase of the third
summand.

\begin{theorem} %
  \label{thmProgress}
  Every rule application in $L^\#$ increases the norm $N(\Obs)$ in \eqref{norm}.
\end{theorem}
The norm $N(\Obs)$ and therefore also the number of rule applications is bounded:
\begin{theorem}
	\label{thmBound}
  If $\Obs$ is an observation tree for $\M$ with $n$ equivalence classes of
  states and $|I| = k$, then
  \(
    N(\Obs) \le
    \frac{1}{2} \cdot n \cdot (n+1) + k n + (n-1)(kn+1)
    \in \bigO(k n^2).
  \)
\end{theorem}
At any point of execution, either rule (R1), (R2), or (R4) is
applicable, so $\lsharp$ never blocks. As soon as the norm $N(\Obs)$ hits the
bound, the only applicable rule is rule (R4) with the teacher
accepting the hypothesis. Thus, the correct Mealy machine is
learned within $\bigO(k\cdot n^2)$ rule applications.
The complexity in terms of the input parameters is studied in \autoref{secComplexity}.

We now continue defining the subroutines and proving them correct.

\subsection{Consistency checking}
\label{secConsistency}
A hypothesis $\Hyp$ is not necessarily \emph{consistent} with $\Obs$, in the
sense of a functional simulation $\Obs\to\Hyp$.
Via a breadth-first search of the Cartesian product of $\Obs$ and $\Hyp$
(\autoref{alg:consistency}), we may check in time linear in the size of $\Obs$
whether a functional simulation $\Obs\to\Hyp$ exists. In the negative case, we
obtain $\sigma\in I^*$ leading to a conflict without any equivalence or output
query to the teacher needed. Thus, this is also called \textqt{counterexample
milking}~\cite{BalcazarDG97}.
\begin{algorithm}[t]
	\begin{algorithmic}
		\Procedure{CheckConsistency}{$\Hyp$}
		\State $Q \gets \code{new}~queue \subseteq S\times S$ 
		\State $enqueue(Q, (q^{\Obs}_0, q^{\Hyp}_0)))$

		\While{$(q,r) \gets dequeue(Q)$\hspace*{4pt}}
    \StateIf{$q\apart r$}
      {\Return \code{no:} $\mathsf{access}(q)$}

		\ForAll{$q\xrightarrow{i/o} p$ in $\Obs$}
      \State $enqueue(Q, (p, \delta^{\Hyp}(r,i)))$
		\EndFor
		\EndWhile
		\State \Return \code{yes}
		\EndProcedure
	\end{algorithmic}
	\caption{Check if hypothesis $\Hyp$ is consistent with observation tree $\Obs$}
	\label{alg:consistency}
\end{algorithm}
\begin{lemma}
  \label{algConsistencyCorrect}
  \autoref{alg:consistency} terminates and is correct, that is, if~$\Hyp$ is a
  hypothesis for $\Obs$ with a complete basis, then
  $\CheckConsistency(\Hyp)$
  \begin{enumerate}[topsep=1pt]
  \item returns $\code{yes}$, if $\Hyp$ is consistent, 
  \item returns $\code{no}$ and $\rho\in I^*$, if
    $\rho$ leads to a conflict ($\delta^{\Obs}(q_0^\Obs,\rho)\apart\delta^{\Hyp}(q_0^\Hyp,\rho)$ in $\Obs$).
  \end{enumerate}
\end{lemma}

\subsection{Counterexample processing}
\fvnote{}
\fvnote{}
The $L^*$ algorithm~\cite{Ang87} performs $\bigO(m)$ queries to analyze a
counterexample of length $m$. So if a teacher returns really long
counterexamples, their analysis will dominate the learning process.  Rivest \&
Schapire \cite{RivestS89,RivestS93} improve counterexample analysis of $L^*$
using binary search, requiring only $\bigO(\log m)$ queries. A similar trick is applied in $\lsharp$.

Suppose $\sigma$ leads to a conflict $q\apart r$ for
$q = \delta^{\Hyp}(q_0^{\Hyp}, \sigma)$ and $r = \delta^{\Obs}(q_0^{\Obs},
\sigma)$. Then, $\ProcessCounterexample(\sigma)$ (\autoref{alg:counterexamplebs}) extends $\Obs$ such
that $\Hyp$ will never be a hypothesis for $\Obs$ again.

\jrnote{}
If $r\in S\cup F$, then the conflict $q\apart r$ is obvious and $\Hyp$ is not a
hypothesis again.
If otherwise $r \not\in S \cup F$,
the binary search will successively reduce the number of transitions of $\sigma$ outside $S\cup F$ by a factor of 2 until we reach the above base case $S\cup F$.
Let $\sigma_1\,\sigma_2 := \sigma$ such that the run of $\sigma_1$ in $\Obs$ ends
halfway between the frontier and $r$. By an additional output query, the binary
search checks whether already $\sigma_1$ leads to a conflict. In the two cases,
we can either avoid $\sigma_1$ or $\sigma_2$, so we reduce the number of
transitions outside $S\cup F$ to half the amount. The precise argument is in:

\begin{algorithm}[t]%
  \begin{minipage}{.65\textwidth}
	\begin{algorithmic}
		\Procedure{\ProcessCounterexample}{$\Hyp$, $\sigma \in I^*$}
		\State $q \gets \delta^{\Hyp}(q_0^{\Hyp}, \sigma)$
		\State $r \gets \delta^{\Obs}(q_0^{\Obs}, \sigma)$
		\If{$r \in S \cup F$}
		\State \Return
		\Else
		\State $\rho \gets$ unique prefix of $\sigma$ with $\delta^{\Obs}(q_0^{\Obs}, \rho) \in  F$
		\State $h \gets \lfloor \frac{| \rho | +| \sigma |}{2} \rfloor$
		\State $\sigma_1 \gets \sigma[1..h]$
		\State $\sigma_2 \gets \sigma[h+1 .. | \sigma |]$
		\State $q' \gets \delta^{\Hyp}(q_0^{\Hyp}, \sigma_1)$
		\State $r' \gets \delta^{\Obs}(q_0^{\Obs}, \sigma_1)$
		\State $\eta \gets$ witness for $q \;\#\; r$
		\State $\OutputQuery(\mathsf{access}(q') \; \sigma_2 \; \eta)$
		\If{$q' \;\#\; r'$}
		\State \ProcessCounterexample($\Hyp$, $\sigma_1$)
		\Else
		\State \ProcessCounterexample($\Hyp$, $\mathsf{access}(q') \; \sigma_2$)
		\EndIf
		\EndIf
		\EndProcedure
	\end{algorithmic}
  \end{minipage}%
  \hfill%
  \begin{minipage}{.33\textwidth}
    \centering
        \begin{tikzpicture}[
      input word/.style={
        decorate,
        decoration={coil,aspect=0,amplitude=1pt},
        draw=white,
        double=black,
        double distance=1pt,
        line width=1pt,
      },
      vertex/.style={
        shape=circle,
        inner sep=.5pt,
        fill=white,
      },
      ]
    \coordinate (root) at (0,0);
    \node[outer sep=0mm,anchor=south west]
          at (root) {$\Obs$};
    \pgfmathsetmacro{\treeAngle}{20} %
    \pgfmathsetmacro{\sigmaAngle}{8} %
    \def\braceDistance{1pt} %
    \def\basisRadius{3.5cm} %
    \def\frontierRadius{4cm} %
    \def\obsTreeHeight{4.2cm} %
    \def\suffixlen{3.8cm} %
    \draw[basis,draw=none] (root) -- +(-90-\treeAngle:\basisRadius)
         arc (-90-\treeAngle:-90+\treeAngle:\basisRadius)
         -- cycle;
    \draw[frontier,draw=none]
         ($ (root) +  (-90-\treeAngle:\basisRadius)    $)
         arc (-90-\treeAngle:-90+\treeAngle:\basisRadius)
         -- ++(-90+\treeAngle:\frontierRadius-\basisRadius)
         arc (-90+\treeAngle:-90-\treeAngle:\frontierRadius)
         -- cycle;

    \draw (root) -- +(-90+\treeAngle:\obsTreeHeight)
          (root) -- +(-90-\treeAngle:\obsTreeHeight);
    \draw[dotted]
          ($ (root) + (-90+\treeAngle:\obsTreeHeight) $) -- +(-90+\treeAngle:7mm)
          ($ (root) + (-90-\treeAngle:\obsTreeHeight) $) -- +(-90-\treeAngle:7mm);

    \begin{scope}[
      mybrace/.style={
        decorate,
        decoration={
          brace,
          mirror,
          raise=5pt,
          amplitude=4pt,
        },
        every label/.append style={
          pos=0.5,sloped,
          text depth=32pt, %
        },
      },
      ]
    \draw [mybrace]
    (root) -- node[every label] {basis}
    (-90-\treeAngle:\basisRadius-\braceDistance)
    ;
    \draw [mybrace]
    ($ (root) + (-90-\treeAngle:\basisRadius+\braceDistance) $)
    -- node[every label] {frontier}
    (-90-\treeAngle:\frontierRadius)
    ;
    \end{scope}

    \coordinate (rho) at (-90+\sigmaAngle: .5*\basisRadius + .5*\frontierRadius);
    \coordinate (p) at ($ (rho) + (-90+\sigmaAngle: \suffixlen)$);
    \coordinate (rp) at ($ (rho) !.5! (p)$);
    \coordinate (qp) at (-90-\sigmaAngle: .88*\basisRadius);
    \coordinate (qp plus sigma) at ($ (qp) + (-90-\sigmaAngle: .6*\suffixlen)$);

    \draw[input word] (root) -- node[pos=0.5] {\treeNodeLabel{$\rho$}} (rho);
    \draw[input word] (rho) -- node[xshift=1mm,pos=0.52] {\treeNodeLabel{$\sigma[|\rho|\mathord{+}1..h]$}} (rp);
    \draw[input word] (rp) -- node[pos=0.52] {\treeNodeLabel{$~\sigma_2$}} (p);
    \draw[input word] (root) --  node[sloped,pos=0.6] {\treeNodeLabel{$\access\,q'$}}(qp);
    \draw[input word] (qp) -- node[pos=0.5] {\treeNodeLabel{$~\sigma_2$}} (qp plus sigma);

    \node[vertex] (rho node) at (rho) {$\bullet$};
    \node[vertex,label=right:{\treeNodeLabel{$r$}}] (p node) at (p) {$\bullet$};
    \node[vertex,label=right:{\treeNodeLabel{$r'$}}] (rp node) at (rp) {$\bullet$};
    \node[vertex] (root node) at (root) {$\bullet$};
    \node[vertex,label={[inner sep=0pt]left:{\treeNodeLabel{$q'$}}}] (qp node) at (qp) {$\bullet$};
    \node[vertex] (qp plus sigma node) at (qp plus sigma) {$\bullet$};
    \end{tikzpicture}
  \end{minipage}
  \caption{Processing $\sigma$
    that leads to a conflict,
    i.e.~$\delta^\Hyp(q_0,\sigma)\apart \delta^\Obs(q_0,\sigma)$
  }
  \label{alg:counterexamplebs}
\end{algorithm}%

\begin{lemma}
  \label{algProcCountCorrect}
  Suppose basis $S$ is complete, $\Hyp$ is a complete Mealy machine containing the
  basis\twnote{}, and $\sigma\in I^*$ leads
  to a conflict.
  Then $\ProcessCounterexample(\Hyp, \sigma)$ terminates, performs at most $\bigO(\log_2|\sigma|)$ output queries and is correct: upon termination, the machine $\Hyp$ is not a hypothesis for $\Obs$ anymore.
\end{lemma}

\subsection{Adaptive distinguishing sequences}
\label{secAds}

\jrnote{}
As an optimization in practice, we may extend the rules (R2) and (R3) by incorporating \emph{adaptive distinguishing sequences} ($\Ads$) into the respective output queries.
Adaptive distinguishing sequences, which are commonly used in the area of conformance testing \cite{LYa94}, are input sequences where the choice of an input may depend on the outputs received in response to previous inputs. 
Thus, strictly speaking, an ADS is a decision graph rather than a sequence.
This mild extension of the learning framework reflects the actual black box
behaviour of Mealy machines: for every input in $I$ sent to the hidden Mealy
machine, the learner observes the output $O$ before sending the next input symbol.
Use of adaptive distinguishing sequences may reduce the number of output queries that are required for the identification of frontier states.

As an example, consider the observation tree of Figure~\ref{Fig:Obsbegin}(left). The basis for this tree consists of $5$ states, which are pairwise apart (separating sequences are $a$, $ab$ and $aa$). Frontier states can be identified by the \emph{single} adaptive sequence of Figure~\ref{Fig:Obsbegin}(right). The ADS starts with input $a$. If the response is $2$ we have identified our frontier state as $t_4$. If the response is $0$ then the frontier state is either $t_0$ or $t_2$, and we may identify the state with a subsequent input $a$. Similarly, if the response is $1$ then the frontier state is either $t_1$ or $t_3$, and we may identify the state by a subsequent input $b$. We can therefore identify (or isolate) frontier state $t_5$ with a single (extended) output query that starts with the access sequence for $t_5$ ($bbbba$) followed by the ADS of Figure~\ref{Fig:Obsbegin}(right).
If we used separating sequences, we would need at least 2 output queries.
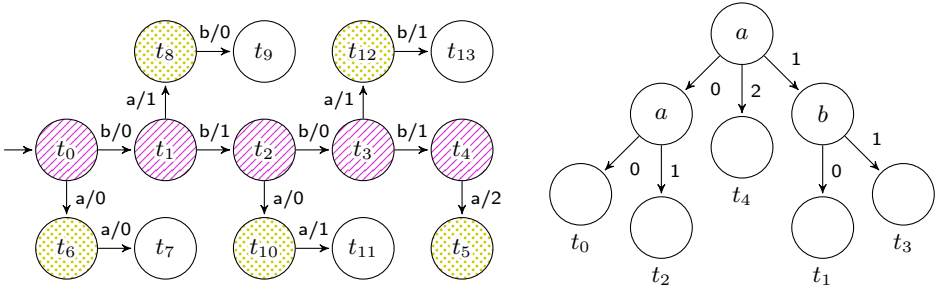
\begin{figure}[h]
		\begin{subfigure}{.60\textwidth}
		\begin{tikzpicture}[->,>=stealth',shorten >=1pt,auto,node distance=1.3cm,main node/.style={circle,draw,font=\sffamily\large\bfseries}]
		\node[initial, state,basis] (0) {\treeNodeLabel{$t_0$}};
		\node[state,basis] (1) [right of=0] {\treeNodeLabel{$t_1$}};
		\node[state,basis] (2) [right of=1] {\treeNodeLabel{$t_2$}};
		\node[state,basis] (3) [right of=2] {\treeNodeLabel{$t_3$}};
		\node[state,basis] (4) [right of=3] {\treeNodeLabel{$t_4$}};
		\node[state,frontier] (5) [below of=4] {\treeNodeLabel{$t_5$}};
		\node[state,frontier] (6) [below of=0] {\treeNodeLabel{$t_6$}};
		\node[state] (7) [right of=6] {\treeNodeLabel{$t_7$}};
		\node[state,frontier] (8) [above of=1] {\treeNodeLabel{$t_8$}};
		\node[state] (9) [right of=8] {\treeNodeLabel{$t_9$}};
		\node[state,frontier] (10) [below of=2] {\treeNodeLabel{$t_{10}$}};
		\node[state] (11) [right of=10] {\treeNodeLabel{$t_{11}$}};
		\node[state,frontier] (12) [above of=3] {\treeNodeLabel{$t_{12}$}};
		\node[state] (13) [right of=12] {\treeNodeLabel{$t_{13}$}};
		
		\path[every node/.style={font=\sffamily\scriptsize}]
		(0) edge node {b/0} (1)
			edge node {a/0} (6)
		(1) edge node {b/1} (2)
			edge node {a/1} (8)	
		(2) edge node {b/0} (3)
			edge node {a/0} (10)
		(3) edge node {b/1} (4)
			edge node {a/1} (12)
		(4) edge node {a/2} (5)
		(6) edge node {a/0} (7)
		(8) edge node {b/0} (9)
		(10) edge node {a/1} (11)
		(12) edge node {b/1} (13);
		\end{tikzpicture}
		\end{subfigure}
		\begin{subfigure}{.40\textwidth}
			\centering
			\begin{tikzpicture}[->,>=stealth',shorten >=1pt,auto,node distance=1.5cm,main node/.style={circle,draw,font=\sffamily\large\bfseries}]
			\node[state] (a) {\treeNodeLabel{$a$}};
			\node[state] (aa) [below left of=a] {\treeNodeLabel{$a$}};
			\node[state] (ab) [below right of=a]{\treeNodeLabel{$b$}};
			\node[state,label={below:$t_4$}] (t4) [below of=a] {\treeNodeLabel{$~$}};
			\node[state,label={below:$t_0$}] (t0) [below left of=aa] {\treeNodeLabel{$~$}};
			\node[state,label={below:$t_2$}] (t2) [below of=aa] {\treeNodeLabel{$~$}};
			\node[state,label={below:$t_1$}] (t1) [below of=ab] {\treeNodeLabel{$~$}};
			\node[state,label={below:$t_3$}] (t3) [below right of=ab] {\treeNodeLabel{$~$}};
			
			\path[every node/.style={font=\sffamily\scriptsize}]
			(a) edge node {0} (aa)
				edge node {1} (ab)
				edge node {2} (t4)
			(aa) edge node {0} (t0)
				 edge node {1} (t2)
			(ab) edge node {0} (t1)
				 edge node {1} (t3);
			\end{tikzpicture}
		\end{subfigure}
	\caption{An observation tree (left) and an ADS for its basis (right)}
	\label{Fig:Obsbegin}
\end{figure}

In the setting of $\lsharp$, we can directly compute an optimal ADS from the current
observation tree. To this end, we recursively define an \emph{expected reward}
function $E$, which sends a set $U\subseteq Q^{\Obs}$ of states to the maximal expected
number of apartness pairs (in the absence of unexpected outputs).
\begin{equation}
  E(U) = \max_{i\in \mathit{inp}(U)} \left(
    \sum_{\substack{o\in O}} \frac{|U\xrightarrow{i/o}| \cdot
      (|U\xrightarrow{i}| - |U\xrightarrow{i/o}|+ E(U\mathord{\xrightarrow{i/o}}))}{|U\xrightarrow{i}|} 
  \right)
  \label{defE}
\end{equation}
where $\mathit{inp}(U) :=  \{i\in I \mid \exists q \in U :  \delta^{\Obs}(q,i)\converges\,\}$,
$U\mathord{\xrightarrow{i}} := \{q\in U \mid \delta^{\Obs}(q,i)\converges\,\}$ and
$U\mathord{\xrightarrow{i/o}} := \{q'\in Q^{\Obs} \mid \exists q\in U : q\xrightarrow{i/o} q'\}$.
We define the maximum over the empty set to be $0$.
Then $\Ads(U)$ is the decision tree constructed as follows:
\begin{itemize}
	\item 
	If $U\mathord{\xrightarrow{i}} = \emptyset$ then $\Ads(U)$ consists of a single node $U$ without a label.
	\item 
		If $U\mathord{\xrightarrow{i}} \neq \emptyset$ then $\Ads(U)$ is constructed by choosing an input $i$ that witnesses the maximum $E(U)$, creating a node $U$ with label $i$, and, for each output $o$ with $U\mathord{\xrightarrow{i/o}} \neq \emptyset$,
		adding an $o$-transition to $ADS(U\mathord{\xrightarrow{i/o}})$.
\end{itemize}
For the observation tree of Figure~\ref{Fig:Obsbegin}(left) we may compute
$E(\{ t_0,\ldots, t_4 \}) = 4$ and obtain the decision tree of Figure~\ref{Fig:Obsbegin}(right) as ADS. Running the ADS from state $t_5$ will create 4 new apartness pairs with basis states (or 5 in case an unexpected output occurs, e.g.\ $a(1) b(2)$).
\takeout{} %

\begin{proposition}
  \label{adscorrect}
  Define $\lsharp_{\Ads}$ by replacing the output queries in $\lsharp$ with
  \begin{description}
  \item[(R2')] $\OutputQuery(\access(q)\;i\;\Ads(S))$ in (R2) and
  \item[(R3')] $\OutputQuery(\access(q)\;\Ads(\{b\in S\mid \neg(b\apart q)\}))$ in (R3).
  \end{description}
  Then, $\lsharp_{\Ads}$ lets the norm $N(\Obs)$ grow for each rule application
  and thus is correct.
\end{proposition}

\subsection{Complexity}
\label{secComplexity}

Since equivalence queries are costly in practice and since processing of long
counterexamples of length $m$ requires $\bigO(\log m)$ output queries, it makes
sense to postpone equivalence queries as long as possible:

\begin{definition}
  \emph{Strategic $\lsharp$} (resp.~$\lsharp_\Ads$) is the special case of \autoref{alg:lsharp} where
  rule (R4) is only applied if none of the other rules is applicable.

\end{definition} 
Then we obtain the following query complexity for the $\lsharp$ algorithm.

\begin{theorem}
	\label{thmQueryComplexity}
  Strategic $\lsharp$ (resp.~$\lsharp_\Ads$) learns the correct
  Mealy machine within $\bigO(kn^2+n \log m)$ output queries and at most $n-1$
  equivalence queries.
\end{theorem}

The query complexity of $\lsharp$ equals the best known query complexity for active learning algorithms, as achieved by
Rivest \& Schapire's algorithm \cite{RivestS89,RivestS93}, the observation pack
algorithm \cite{ThesisFalk}, the TTT algorithm
\cite{Isberner2014,Isberner15}, and the ADT algorithm \cite{Frohme19}.

In a black box learning setting in practice, answering an output query for
$\sigma \in I^*$ grows linearly with the length $\sigma$. Therefore, the
(asymptotic) total number of input symbols sent by the learner is also a metric
for comparing learning algorithms:
\begin{theorem}
	\label{thmSymbolComplexity}
	Let $n \in \bigO(m)$. Then
	the strategic $\lsharp$ algorithm learns the correct Mealy machine with
	$\bigO(kmn^2+nm \log m)$ input symbols.
\end{theorem}

This matches the asymptotic symbol complexity of the best known
active learning algorithms. 
Although \ProcessCounterexample\ reduces the length of the sequence leading to the
conflict, the witness of the conflict remains of size $\Theta(m)$ in the worst case.
This means that we need $\bigO(m \log m)$ symbols to process a single counterexample and  $\bigO(n m \log m)$ symbols to process all counterexamples.

\section{Experimental Evaluation}
\label{sec:bench}
In the previous sections, we have introduced and discussed the $\lsharp$ algorithm.
We now present a short experimental evaluation of the algorithm to
demonstrate its performance when compared to other state-of-art algorithms.
We run two versions of $\lsharp$: the base version (\autoref{alg:lsharp}),
and the ADS optimised variant ($\lsharp_{\Ads}$), and compare these with the (highly optimized) LearnLib\footnote{\url{https://learnlib.de/}} implementations of
TTT, ADT,\footnote{The ADT algorithm makes use of some heuristics to guide the learning process,
  we have selected the ``Best-Effort'' settings.}
and `RS', by which we refer to $L^{*}$ with Rivest-Schapire counterexample processing~\cite{RivestS89,RivestS93}.
All source-code and data is available online.\footnote{
  \url{https://gitlab.science.ru.nl/sws/lsharp} and
  \doiurl{10.5281/zenodo.5735533}}

\paragraph{Implementing \EquivalenceQuery:}
We implement equivalence queries using conformance testing,
which also makes output queries.
We have fixed the testing tool to
Hybrid-ADS\footnote{\url{https://github.com/Jaxan/hybrid-ads}} \cite{SmeenkMVJ15}.
Hybrid-ADS has multiple
configuration options, and we have set the state cover mode to ``buggy'', the
number of extra states to check for to 10, the number of infix symbols to 10,
and the mode of execution to ``random'', generating an infinite test-suite.
Note that with these settings, the equivalence queries are not exact in general but approximated via random testing.

\paragraph{Data-set and metrics:}
We use a subset of the models available from the AutomataWiki (see \cite{NeiderSVK97}):
we learn models for the SSH, TCP, and TLS protocols, alongside the BankCard models.
The largest model in this subset has 66 states and 13 input symbols.
We record the number of output queries and input symbols used during learning and
testing, alongside the number of equivalence queries required to learn each model.
An output query is a sequence $\sigma\in I^*$ of $|\sigma|$ input symbols and one \emph{reset} symbol.
A reset symbol returns the \emph{system under test} (SUT) to its initial state.
So \emph{resets} denotes the number of output queries and
\emph{inputs} denotes the total number of symbols sent to the SUT.
We believe that these metrics accurately portray the effort required to learn a model.

\paragraph{Experiment Set-up:}
All experiments were run on a Ryzen 3700X processor with 32GB of memory, running
Linux.
Each experiment refers to completely learning a model of the SUT.
Due to the effects of randomization in the equivalence oracle, we repeat each
experiment 100 times.

\paragraph{Results and Discussion}
\begin{figure}[h]
   \begin{subfigure}{.5\textwidth}
    \includegraphics[draft=false,width=\textwidth]{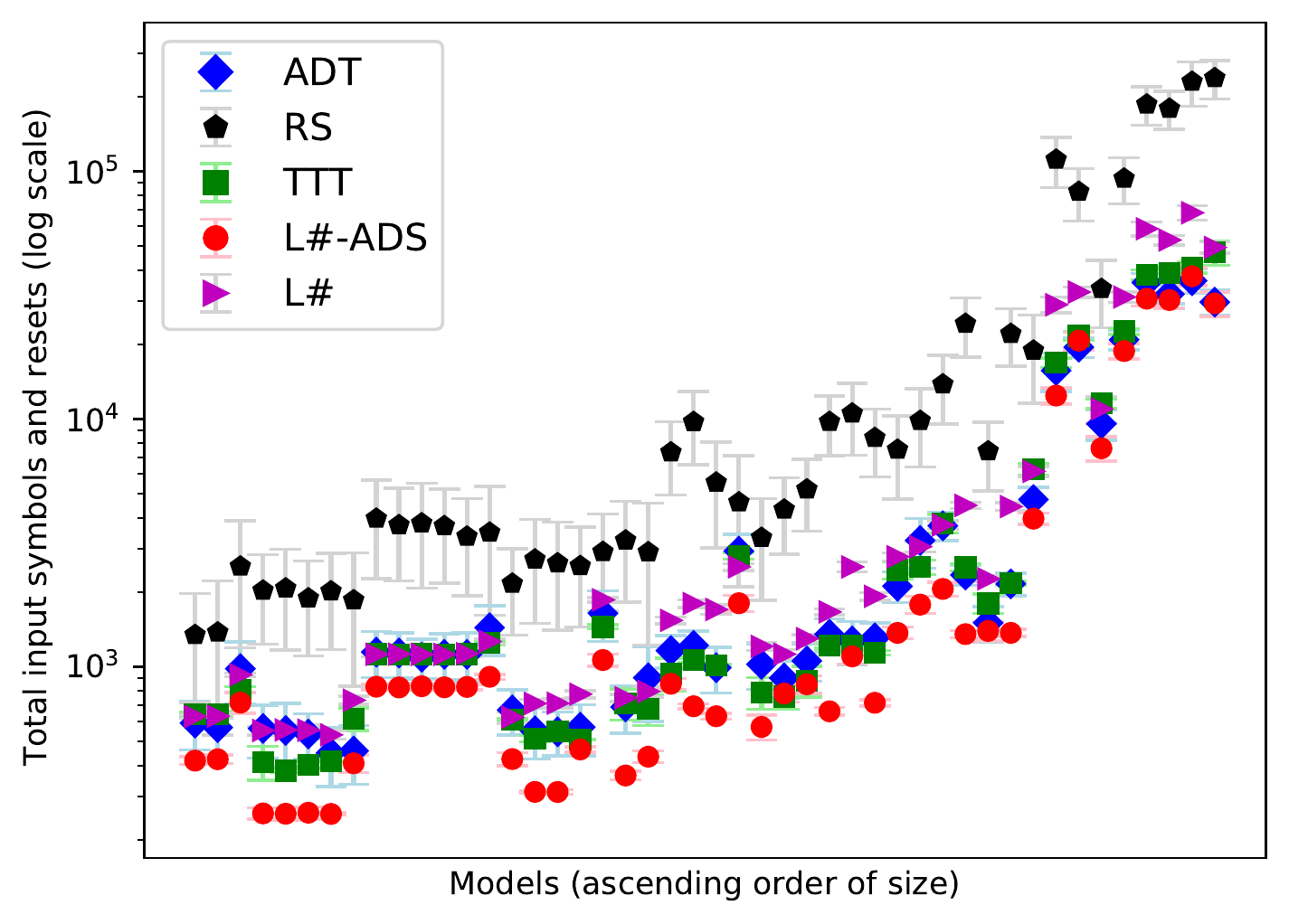}%
    \caption{Symbols used during learning phase}
    \label{fig:learning_metric_plot}
   \end{subfigure}%
  \begin{subfigure}{.5\textwidth}%
    \includegraphics[draft=false,width=\textwidth]{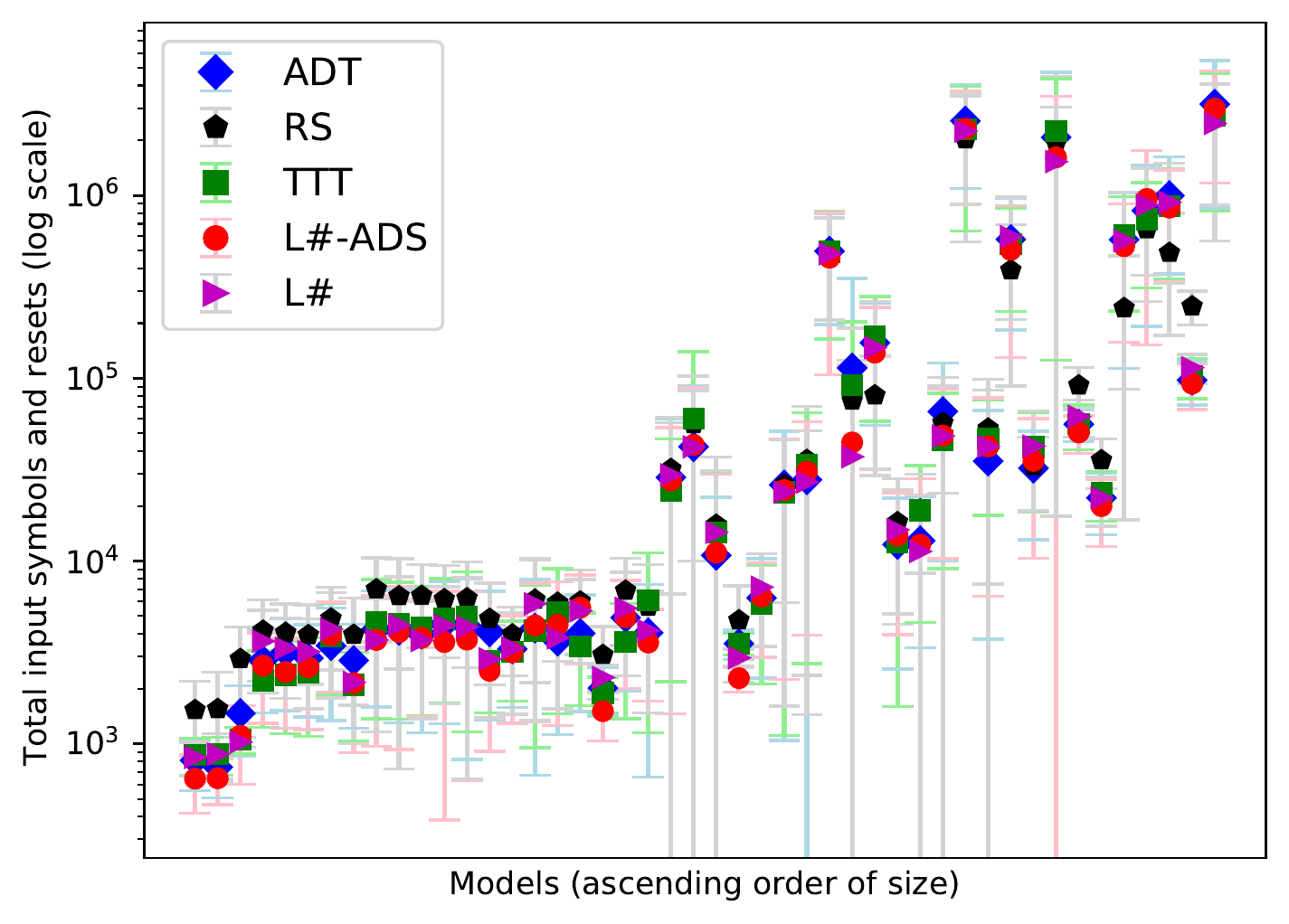}
    \caption{Symbols used both learning and testing}
    \label{fig:total_metric_plot}
  \end{subfigure}
  \caption{Performance plots of the selected learning algorithms (lower is better.)
}
\end{figure}

\autoref{fig:learning_metric_plot} shows the total size of data sent by the
learning algorithms via output queries -- so both the number and the size of
output queries are counted. In order to incorporate the equivalence
queries, \autoref{fig:total_metric_plot} shows the total size of data sent to
the SUT during learning and testing.
Note, in both plots the y-axis is log-scaled.
The x-axis indicates the models, sorted in increasing number of states.
The bars indicate standard deviation.

We can observe from the learning phase plot (\autoref{fig:learning_metric_plot}) that
$\lsharp$ expectedly does not perform better than the TTT and ADT algorithms, while
the RS algorithm performs the worst among all four.
However, $\lsharp_{\Ads}$ usually performs better than -- or, at least, is competitive with --
ADT and TTT.
Furthermore, the error bars in the learning phase are very small, indicating that the measurements
are stable.
Generally, depending on the models a
different algorithm is the fastest, but for every model,
$\lsharp_{\Ads}$ is among the fastest, with and without the exclusion of the testing phase.

\autoref{fig:total_metric_plot} presents the total number of input symbols and resets sent to the SUT.
All algorithms seem to be very close
in performance, which may be explained by the testing phase dominating the process.
Indeed, Aslam et al.~\cite{AslamCSB20} experimentally demonstrated that it is largely the testing phase
which influences learning effort.

The complete benchmark results (\ifthenelse{\boolean{showappendix}}{%
  in \appendixref{completeBenchmarks}%
}{%
  in the appendix of \cite{VGRW22}%
})
show more detailed information of the learned models, and
highlights the smallest number per column and model.
We can see that the number of equivalence queries are roughly similar for almost all the algorithms,
while $\lsharp$ seems to perform better for some models in the learning phase.

\section{Conclusions and Future Work}
\label{sec:concl}

We presented $\lsharp$, a new algorithm for the classical problem of active automata learning. The key idea behind the approach is to focus on establishing \emph{apartness}, or inequivalence of states, instead of approximating equivalence as in $L^*$ and its descendants. Concretely, the table/discrimination tree in $L^*$-like algorithms is replaced in $\lsharp$ by an observation tree, together with an apartness relation. 
This change in perspective leads to a simple but effective algorithm, which reduces the total number of symbols required for learning when compared to state-of-the-art algorithms. In particular, the use of observation trees, which are essentially tree-shaped Mealy machines, enables a modular integration of testing techniques, such as the ADS method, to identify states.
Although the asymptotic output query complexity of $\lsharp$ is $\bigO(kn^2+n \log m)$, in our experiments $\lsharp$ only needs in between $kn$ and $4kn$ output queries (resets) to learn the benchmark models (with $n \leq 66)$, which means that on average $\lsharp$ needs in between 1 and 4 output queries to identify a frontier state.

Of course there are also similarities between $\lsharp$ and $L^*$.
The basis of $\lsharp$ is comparable to the top half of the $L^*$ table: both in $\lsharp$ and in (\cite{RivestS93}'s version of) $L^*$ these prefixes induce a spanning tree. The frontier of $\lsharp$ is comparable to the bottom half of the $L^*$ table.
But whereas $L^*$ constructs residual classes of the language, $\lsharp$  builds an automaton directly from the observation tree. As a consequence, $L^*$ asks redundant queries, and optimizations of $L^*$ try to avoid this redundancy. In contrast, $\lsharp$ does not even think about asking redundant queries since it operates directly on the observation tree and only poses queries that increase the norm.

There is still much work to do to improve our prototype implementation, to include additional conformance testing algorithms, and to extend the experimental evaluation to a richer set of benchmarks and algorithms.
One issue that we need to address is scaling of $\lsharp$ to bigger
models. Our prototype implementation easily learns Mealy machines with hundreds of states, but fails to learn larger models such as the ESM benchmark of \cite{SmeenkMVJ15} (3410 states, 78 inputs) because the observation tree becomes too big ($\approx$25 million nodes will be required for the ESM). We see several ways to address this issue, e.g.,
pruning the observation tree, only keeping short ADSs to separate the basis states, storing parts of the tree on disk, distributing the tree over multiple processors (parallelizing the learning process), and using existing platforms for big graph processing \cite{biggraph21}.

Aslam et al.\ \cite{AslamLSB18} report on experiments in which active learning techniques are applied to 202 industrial software components from ASML. Out of these, interface protocols could be successfully derived for 134 components (within a give time bound). One of the main conclusions of the study is that the equivalence checking phase (i.e.\ conformance testing of hypothesis models) is the bottleneck for scalability in industry. We believe that a tighter integration of learning and testing, as enabled by $\lsharp$, will be key to address this challenging problem. 

It will be interesting to extend $\lsharp$ to richer frameworks such as register automata, symbolic automata and weighted automata. In fact, we discovered $\lsharp$ while working on a grey-box learning algorithm for symbolic automata.

\label{maintextend} %
\bibliographystyle{splncs04}
\bibliography{abbreviations,dbase}

\vfill

{\small\medskip\noindent{\bf Open Access} This chapter is licensed under the terms of the Creative Commons\break Attribution 4.0 International License (\url{http://creativecommons.org/licenses/by/4.0/}), which permits use, sharing, adaptation, distribution and reproduction in any medium or format, as long as you give appropriate credit to the original author(s) and the source, provide a link to the Creative Commons license and indicate if changes were made.}

{\small \spaceskip .28em plus .1em minus .1em The images or other third party material in this chapter are included in the chapter's Creative Commons license, unless indicated otherwise in a credit line to the material.~If material is not included in the chapter's Creative Commons license and your intended\break use is not permitted by statutory regulation or exceeds the permitted use, you will need to obtain permission directly from the copyright holder.}

\medskip\noindent\includegraphics{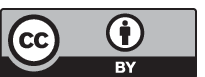}

\ifthenelse{\boolean{showappendix}}{
  \newpage
  \appendix
  
\section{Omitted Proofs}
\label{proofs}
\begin{proofappendix}{la: apartness refinement}
  For the witness $\sigma\vdash q\apart q'$, we have
  $\lambda^{\Obs}(q,\sigma)\converges$ and $\lambda^{\Obs}(q',\sigma)\converges$. 
  Since $f$ is a functional simulation, 
  by \autoref{la:refinement}
  we have $\lambda^{\M}(f(q), \sigma) \converges$ and $\lambda^{\Obs}(q,\sigma) = \lambda^{\M}(f(q),\sigma)$,
  and similarly $\lambda^{\Obs}(q',\sigma) = \lambda^{\M}(f(q'),\sigma)$.
  Hence,
  \[
    \lambda^{\M}(f(q),\sigma)
    = \lambda^{\Obs}(q, \sigma)
    \neq \lambda^{\Obs}(q', \sigma)
    = \lambda^{\M}(f(q),\sigma)
  \]
  which proves $\semantics{f(q)}\neq \semantics{f(q')}$.
  \jrnote{}
\end{proofappendix}

\begin{proofappendix}{la: weak co-transitivity}
  The witness
  $\sigma\vdash r\apart r'$ implies that $\lambda(r,\sigma)\converges$, $\lambda(r',\sigma)\converges$,
  and $\lambda(r,\sigma) \neq \lambda(r',\sigma)$.
  Since $\lambda(q,\sigma)\converges$, $\neg(r\apart q)\wedge
  \neg(r'\apart q)$ leads to the contradiction
  \[
    \lambda(r,\sigma)
    \overset{\neg(r\apart q)}{=}
    \lambda(q,\sigma)
    \overset{\neg(r'\apart q)}{=}
    \lambda(r',\sigma).
    \tag*{\qedhere}
  \]
\end{proofappendix}

\begin{proofappendix}{hypExistence}
  Consider the relation:
  \[
    \Delta \subseteq (S\times I) \times S
    \qquad
    ((q,i),q')
    \in \Delta
    \quad
    \text{if}
    \quad
    \delta^{\Obs}(q,i)\diverges
    ~\text{ or }~
    \neg(\delta^{\Obs}(q,i) \apart q')
  \]
  \begin{enumerate}
  \item If there are no isolated states, then every $(q,i)\in S\times I$
    is related to some $q'\in S$, so there is some functional relation
    $\delta^{\Hyp}\subseteq \Delta$ of type $\delta^{\Hyp}\colon S\times I\to S$.
  \item If $S$ is complete and all states in $F$ are identified, then $\Delta$
    is already a functional relation $\delta^{\Hyp} = \Delta\colon S\times I\to S$.
    \qedhere
  \end{enumerate}
\end{proofappendix}

\begin{proofappendix}{hypSize}
In the proof of \autoref{hypSize},
we characterize equivalence of Mealy machines via  \emph{bisimulations}.

\begin{definition}
	\label{def bisimulation MM}
	A \emph{bisimulation} between Mealy machines $\M$ and $\N$ is a relation $R
  \subseteq Q^{\M} \times Q^{\N}$ satisfying, for all $q \in Q^{\M}$, $r \in
  Q^{\N}$, $i \in I$, $o \in O$,
	\begin{equation*}
		q^{\M}_0 \; R \; q^{\N}_0
    \qquad\text{and}\qquad
		q \mathrel{R} r ~\wedge~
    q\xrightarrow{i/o} q'
    ~\Rightarrow ~
                      \exists r'\colon
                      r\xrightarrow{i/o} r'
                      ~\wedge~
                      q'\mathrel{R} r'
	\end{equation*}
	We write $\M \simeq \N$ if there exists a bisimulation relation between $\M$ and $\N$.
\end{definition}

\begin{lemma}
	\label{La:approx}
  Given complete Mealy machines $\M$ and $\N$, the equivalence relation
  $\mathord{\approx} \subseteq Q^{\M}\times Q^{\N}$ is a bisimulation.
\end{lemma}

The next lemma, which is a variation of the classical result of \cite{Pa81}, is again easy to prove.
\begin{lemma}
	\label{la:bisimulation}
	Let $\M$ and $\N$ be complete Mealy machines.
	Then $\M \simeq \N$ iff $\M \approx \N$.
\end{lemma}

We now come to the actual proof of \autoref{hypSize}:

\begin{proof}[of \autoref{hypSize}]
	Let $f$ be a refinement from $\Obs$ to $\M$.  Define relation $R \subseteq S \times Q^{\M}$ by
	\begin{eqnarray*}
		(q, r) \in R & \Leftrightarrow & f(q) \approx^{\M} r.
	\end{eqnarray*}
	We claim that $R$ is a bisimulation between $\Hyp$ and $\M$.
	\begin{enumerate}
		\item 
		Since $f$ is a refinement from $\Obs$ to $\M$, $f(q^{\Obs}_0) = q^{\M}_0$. By construction, $q^{\Obs}_0 = q^{\Hyp}_0$.
		Now the fact that equivalence relation $\approx^{\M}$ is reflexive implies $f(q^{\Hyp}_0) \approx^{\M} q^{\M}_0$, and therefore $(q^{\Hyp}_0, q^{\M}_0) \in R$.
		\item 
		Suppose $(q,r) \in R$ and $i \in I$. Let $q' = \delta^{\Hyp}(q,i)$ and $r' = \delta^{\M}(r,i)$.
		We need to show that $\lambda^{\Hyp}(q,i) = \lambda^{\M}(r,i)$ and $(q', r') \in R$.
		We consider two cases:
		\begin{enumerate}
			\item 
			$\delta^{\Obs}(q,i) \in S$. 
			Then, by construction of $\Hyp$, $\lambda^{\Hyp}(q,i) = \lambda^{\Obs}(q,i)$ and $q' = \delta^{\Obs}(q,i)$.
			Moreover, as $f$ is a refinement from $\Obs$ to $\M$, 
			$f(q')  =  \delta^{\M}(f(q),i)$ and $\lambda^{\Obs}(q,i) = \lambda^{\M}(f(q),i)$. 
			By definition of $R$, $f(q) \approx^{\M} r$. 
			Hence, by Lemma~\ref{La:approx}, $\lambda^{\M}(f(q),i) = \lambda^{\M}(r,i)$ and $\delta^{\M}(f(q),i) \approx^{\M}  r'$.
			By combining the derived equalities we obtain:
			\begin{eqnarray*}
				\lambda^{\Hyp}(q,i) & = & \lambda^{\Obs}(q,i) =  \lambda^{\M}(f(q),i) = \lambda^{\M}(r,i),\\
				f(q') & = &  \delta^{\M}(f(q),i)  \approx{^M}  r'.
			\end{eqnarray*}
			Hence by definition of $R$, $(q', r') \in R$, as required.
			\item 
			$\delta^{\Obs}(q,i) \in F$. Let $q'' = \delta^{\Obs}(q,i) \in F$.
			Then, by construction of $\Hyp$, $\lambda^{\Hyp}(q,i) = \lambda^{\Obs}(q,i)$ and $q'$ is the unique state in $S$ such that $q''$ and $q'$ are not apart.
			By Lemma~\ref{La:approx},
			since all states of $S$ are pairwise apart, all states in the image of $s$ under $f$ are in different equivalence classes of $\approx^{\M}$.			
			Since $\approx^{\M}$ has as many equivalence classes as the number of states of $S$, each state of $\M$ belongs to the same equivalence
			class as $f(s)$, for some $s \in S$.
			Since $q''$ is apart from all states of $S$ except $q'$, $f(q'')$ does not belong to the same equivalence class as $f(s)$, for $s \in S \setminus \{ q' \}$, by Lemma~\ref{La:approx}.
			Hence, by the Sherlock Holmes principle, $f(q'') \approx^{\M} f(q')$.
			Since $f$ is a refinement from $\Obs$ to $\M$, 
			$f(q'')  =  \delta^{\M}(f(q),i)$ and $\lambda^{\Obs}(q,i) = \lambda^{\M}(f(q),i)$. 
			By definition of $R$, $f(q) \approx^{\M} r$. 
			Hence, by Lemma~\ref{La:approx}, $\lambda^{\M}(f(q),i) = \lambda^{\M}(r,i)$ and $\delta^{\M}(f(q),i) \approx^{\M}  r'$.
			By combining the derived equalities we obtain:
			\begin{eqnarray*}
				\lambda^{\Hyp}(q,i) & = & \lambda^{\Obs}(q,i) =  \lambda^{\M}(f(q),i) = \lambda^{\M}(r,i),\\
				f(q') & \approx^{\M} &  f(q'') = \delta^{\M}(f(q),i)  \approx{^M}  r'.
			\end{eqnarray*}
			As equivalence relation $\approx^{\M}$ is transitive, $f(q') \approx^{\M} r'$, and hence by definition of $R$, $(q', r') \in R$, as required.
		\end{enumerate}
	\end{enumerate}
	The theorem now follows by application of Lemma~\ref{la:bisimulation}.
\end{proof}
\end{proofappendix}

\begin{proofappendix}{thmProgress}
  In all cases, let $S,F,\Obs$ denote the values before and $S',F',\Obs'$ denote
  the values after the respective rule application. Also introduce
  abbreviations:
  \[
    N_Q(\Obs) = \frac{|S| \cdot (|S|+1)}{2}
  \]
  \[
    N_{\converges}(\Obs) = \{ (q, i) \in S \times I \mid \delta(q,i) \converges \}
  \]
  \[
    N_{\apart}(\Obs) = \{ (q, q') \in S \times F \mid q \apart q' \}
  \]
  The total norm is:
  \[
    N(\Obs) = N_Q(\Obs) + |N_{\converges}(\Obs)| + |N_{\apart}(\Obs)|
  \]
  \begin{enumerate}
  \item If $q$ is isolated and is thus moved from $F$ to $S$, i.e.~$S' := S\cup
    \{q\}$, then we have
    \begin{align*}
      N_Q(\Obs') &= \frac{|S'| \cdot (|S'|+1)}{2}
      = \frac{(|S|+1) \cdot (|S|+1+1)}{2}
      \\
                 &= \frac{(|S|+1)\cdot |S|}{2} + \frac{(|S|+1)\cdot 2}{2}
                   = N_Q(\Obs) + |S| + 1
    \end{align*}
    \begin{align*}
      N_\downarrow(\Obs') \supseteq N_\downarrow(\Obs)
    \end{align*}
    Finally we have
    \begin{align*}
      N_{\apart}(\Obs') ~~\supseteq~~ N_{\apart}(\Obs) \setminus (S\times \{q\})
    \end{align*}
    and thus
    \[
      |N_{\apart}(\Obs')| ~~\ge~~ |N_{\apart}(\Obs)| - |S|.
    \]
    In total, $N(\Obs') \ge N(\Obs)+1$.

  \item In the second rule, let $\delta^{\Obs}(q,i)$ for some $q\in S$, $i\in I$.
    After the output query for $\mathsf{access}(q) \, i\, \Ads(S)$, we have
    \[
      N_Q(\Obs') = N_Q(\Obs)
      \qquad
      N_\downarrow(\Obs') = N_{\downarrow}(\Obs)\cup\{(q,i)\}
      \qquad
      N_{\apart}(\Obs') \subseteq N_{\apart}(\Obs)
    \]
    and thus $N(\Obs') \ge N(\Obs) + 1$.

  \item For the third rule, consider a state $q\in F$ and distinct $r,r'\in S$
    with $\neg(q\apart r)$ and $\neg(q\apart r')$. The algorithm performs 
    the query
    \[
    \OutputQuery(\mathsf{access}(q) \; \sigma).
    \]
    Hence, $\delta^\Obs(q,\sigma)\converges$ in the updated observation tree, which
    implies $r\apart q$ or $r' \apart q$ by weak
    co-transitivity~(\autoref{la: weak co-transitivity}). Thus,
    \[
    N_{\apart}(\Obs') \supseteq N_{\apart}(\Obs) \cup \{(r,q)\}
    \qquad\text{ or } \qquad 
    N_{\apart}(\Obs') \supseteq N_{\apart}(\Obs) \cup \{(r',q)\}
    \]
    and therefore $|N_{\apart}(\Obs')| \ge |N_{\apart}(\Obs)| + 1$.
    The other components of the norm stay unchanged, thus the norm rises.

  \item If the fourth rule did not terminate the algorithm, we show that $\Hyp$
    is not a hypothesis for $\Obs'$ anymore. By \autoref{algConsistencyCorrect}
    and by \EquivalenceQuery, we have in any case that $\sigma\in I^*$ is such
    that $\delta^{\Hyp}(q_0^\Hyp,\sigma)\apart \delta^{\Obs'}(q_0^{\Obs'},\sigma)$
    (in a possibly extended observation tree $\Obs'$). Moreover, during this rule, in
    \CheckConsistency and \ProcessCounterexample, the basis $S$ is not modified:
    $S= S'$. Even though the observation tree has been updated since \BuildHypothesis, $\Hyp$ still meets the
    criteria of \autoref{algProcCountCorrect}. Hence, after counter example
    processing, $\Hyp$ is not a hypothesis for the updated $\Obs'$ anymore, that is, there
    exist $p\in S$,
    $p\xrightarrow{i/o} q$ in $\Hyp$, and $p\xrightarrow{i/o'}r$ in $\Obs'$
    with $o\neq o'$ or $q\apart r$. But the case $o \neq o'$ does
    not occur: since $S$ is complete, and $\Hyp$ is a hypothesis for $\Obs$,
    the transition $p \xrightarrow{i/o} q$ in $\Hyp$
    implies $p\xrightarrow{i/o}r$ in $\Obs$ and therefore also in the extension $\Obs'$. 
    Hence $q\apart r$, so that $r \neq q$ and consequently $r\in F$. Therefore, we obtain:
      \[
        (q,r) \in N_{\apart}(\Obs') \setminus N_{\apart}(\Obs).
        \tag*{\qedhere}
      \]
  \end{enumerate}
\end{proofappendix}

\begin{proofappendix}{thmBound}
  Since there is a functional simulation $\Obs\to \M$, we have by
  \autoref{la: apartness refinement} that
  \begin{eqnarray*}
    | S | \leq n.
  \end{eqnarray*}
  The number of successors of the basis is bounded by $k \cdot n$:
  \begin{eqnarray*}
    | \{ (q, i) \in S \times I \mid \delta(q,i) \converges \} | & \leq & k n
  \end{eqnarray*}
  The set $S \cup F$ (the basis and all its successor states) contains at most $kn
  +1$ elements. Since each state in the frontier can be apart from at most $n-1$ states in the basis, this means we have
  \begin{eqnarray*}
    | \{ (q, q') \in S \times F \mid q \apart q' \} | & \leq & (n-1)(kn +1)
  \end{eqnarray*}
  In $\bigO$-notation this simplifies to
  \[
    N(\Obs) \le
    \frac{1}{2} n (n+1) + k n + (n-1)(kn+1)
    \in \bigO(k n^2)
    \tag*{\qedhere}
  \]
\end{proofappendix}

\begin{proofappendix}{algConsistencyCorrect}
  The breadth-first search in \autoref{alg:consistency} verifies whether there
  is a functional simulation $f\colon \Obs\to \Hyp$. Since $\Hyp$ is
  deterministic (like all Mealy machines considered here) and since every state
  of $\Obs$ is reachable from the root, there is at most one functional
  simulation $\Obs\to \Hyp$. Thus, consistency checking amounts to verifying
  whether the map
  \[
    f\colon Q^{\Obs} \to Q^{\Hyp}
    \qquad
    f(q) := \delta^{\Hyp}(q_0^{\Hyp}, \access(q))
  \]
  is a functional simulation (\autoref{def refinement}).
  \begin{itemize}
  \item If the procedure returns $\code{no}$, then $q \apart f(q)$ for some
    $q\in \Obs$. Note that $f$ is idempotent, because $\Hyp$ contains $S$: $f(q) = f(f(q))$ (using $Q^{\Hyp}\subseteq Q^{\Obs}$). If $f$ was a
    functional simulation $\Obs\to \Hyp$, this would lead to a contradiction: 
    applying \autoref{la: apartness refinement} to $q\apart f(q)$ (in $\Obs$)
    implies that $f(q)\not\approx f(f(q)) = f(q)$ (in $\M$), a contradiction to
    the reflexivity of $\approx$.

  \item If the procedure returns $\code{yes}$, then $\neg (q\apart f(q))$ for
    all $q\in Q^{\Obs}$. For the verification that $f$ is a functional
    simulation, first note that we trivially have $f(q_0^\Obs) = q_0^{\Hyp}$.
    For the preservation of transitions, consider $q\xrightarrow{i/o} p$ in $\Obs$ and
    $f(q)\xrightarrow{i/o'} p'$ in $\Hyp$.
    Since the basis $S$ is complete in $\Obs$, we have $\lambda^{\Obs}(f(q),i) =
    o'$. Thus $o=o'$, because otherwise we had $i\vdash q\apart f(q)$.
    Note that $\access(p) = \access(q)\,i$ and so
    \begin{align*}
      f(p) &= \delta^{\Hyp}(q_0^{\Hyp},\access(p))
      = \delta^{\Hyp}(q_0^{\Hyp},\access(q)\,i)
      = \delta^{\Hyp}(f(q),i)
      = p'
    \end{align*}
    and thus $f$ is a functional simulation.
    \qedhere
  \end{itemize}
\end{proofappendix}

\begin{proofappendix}{algProcCountCorrect}
  We prove termination by providing a bound on the number of recursive
  calls. For an input word $\sigma\in I^*$ with
  $\delta^\Obs(q_0^\Obs,\sigma)\defined$, we define the \emph{distance from the
    frontier} $d(\sigma)\in \Nat$ by:
  \[
    d(\sigma) = |\sigma| - \max \{|\rho|\mid \rho\text{ prefix of }\sigma,
    \delta^\Obs(q_0^\Obs, \rho) \in S\cup F\}
  \]
  Observe that:
  \begin{itemize}
  \item $d(\sigma) = 0$ iff $r:=\delta^\Obs(q_0^\Obs,\sigma) \in S\cup F$.
  \item If $d(\sigma)>0$ then $d(\sigma) = |\sigma|-|\rho| \ge 1$ with $\rho$ defined as
    in \autoref{alg:counterexamplebs}. For the decomposition
    $\sigma=\sigma_1\cdot \sigma_2$, we have
    \[
      d(\sigma_1) = h - |\rho|
      = \left\lfloor \frac{|\rho|+|\sigma|}{2} \right\rfloor - |\rho|
      = \left\lfloor |\rho|+\frac{|\sigma|-|\rho|}{2} \right\rfloor - |\rho|
      = \left\lfloor \frac{d(\sigma)}{2} \right\rfloor.
    \]
    Since $q':=\delta^{\Hyp}(q_0^{\Hyp},\sigma_1)\in S$ by definition, we have
    that $\delta^{\Obs}(q', i) \in S\cup F$ if $i$ is the first character of $\sigma_2$.
    Note that if $\sigma_2$ is empty, then $d(\access(q')\,\sigma_2) = 0 \le
    \dfrac{d(\sigma)}{2}$, trivially. So if $\sigma_2$ is not empty, then we have:
    \begin{align*}
      d(\access(q')\,\sigma_2)
      \,&\le\, |\sigma_2| - 1
      = |\sigma| - h - 1
      = |\sigma| - \left\lfloor \frac{|\rho|+|\sigma|}{2}\right\rfloor - 1
      \\ &
      = \left\lceil |\sigma| - \frac{|\rho|+|\sigma|}{2}\right\rceil - 1
      \le \left\lfloor |\sigma| - \frac{|\rho|+|\sigma|}{2}\right\rfloor
           \\ &
      = \left\lfloor \frac{|\sigma|-|\rho|}{2}\right\rfloor
      \le \left\lfloor \frac{d(\sigma)}{2}\right\rfloor.
           \hspace*{-1cm} %
    \end{align*}
    So in any of the two recursive calls, if $\sigma'\in I^*$ denotes parameter
    passed to the recursive call, then we have $2\cdot d(\sigma') \le
    d(\sigma)$. This implies termination.
  \end{itemize}
Let $\mathit{OQ}(n)$ denote the maximal number of output queries performed during a run of
 \autoref{alg:counterexamplebs} with $n = d(\sigma)$. Then, using the above observations, we may show by induction on $d(\sigma)$ that
 \begin{eqnarray*}
 	\mathit{OQ}(n) & \leq & \begin{cases}
 		0 & \mbox{if } n=0\\
 		\log_2(2n) & \mbox{if } n>0
 	\end{cases}
 	\end{eqnarray*}
  Since $d(\sigma) < | \sigma |$, this implies that the number of output queries is bounded by
  $\bigO(\log(|\sigma |))$.

  For correctness,
  let $q := \delta^{\Hyp}(q_0^{\Hyp}, \sigma)$
  and $r := \delta^{\Obs}(q_0^{\Obs}, \sigma)$
  as in
  \autoref{alg:counterexamplebs}
  such that $\eta\vdash q\apart r$ for some $\eta\in I^*$,
  i.e.~$\lambda^{\Obs}(q,\eta) \neq \lambda^{\Obs}(r,\eta)$
  \begin{itemize}
  \item In the case of $r\in S\cup F$,
    note that since $q_0^{\Hyp} = q_0^{\Obs}$ and
    $q\apart r$, we have $q_n \neq r_n$ and $|\sigma| \ge 1$.
    Hence, we can decompose $\sigma = \alpha\,i$ into $\alpha\in I^*$ and $i\in I$.
    Let $q' = \delta^{\Hyp}(q_0^{\Hyp},\alpha)$ and $r' =
    \delta^{\Hyp}(q_0,^{\Obs},\alpha)$. Since $\Obs$ is a tree,
    we necessarily have $\alpha = \access(r')$. Hence,
    \[
      q' = \delta^{\Hyp}(q_0, \alpha)
      = \delta^{\Hyp}(q_0, \access(r')) = r'.
    \]
    We have $q'\xrightarrow{i} q$ in $\Hyp$ and $r'\xrightarrow{i} r$ in $\Obs$
    with $q'=r'$ but $q\apart r$, hence $\Hyp$ is not a hypothesis for $\Obs$.

  \item Let $\sigma = \sigma_1\,\sigma_2$ be the decomposition into
    $\sigma_1,\sigma_2\in I^*$, and let $q' := \delta^{\Hyp}(q_0^{\Hyp},
    \sigma_1) \in S$ and $r' := \delta^{\Obs}(q_0^{\Obs}, \sigma_1)$ as in
    \autoref{alg:counterexamplebs}.
    After $\OutputQuery$, we have
    $\lambda^{\Obs}(q', \sigma_2\,\eta)\converges$ and thus:
    \begin{enumerate}
    \item If $q'\apart r'$, then $\delta^{\Hyp}(q_0^{\Hyp}, \sigma_1)\apart
      \delta^{\Obs}(q_0^{\Obs}, \sigma_1)$, so $\sigma_1$ is a valid parameter
      to \ProcessCounterexample and shorter than $\sigma$, so by induction,
      $\Hyp$ is not a hypothesis anymore after the recursive call.
    \item If $\neg(q'\apart r')$, then we necessarily have that
      \[
        \lambda^{\Obs}(q', \sigma_2\,\eta)
        = \lambda^{\Obs}(r', \sigma_2\,\eta)
      \]
      and thus also
      \begin{equation}
        \lambda^{\Obs}(\delta^{\Obs}(q',\sigma_2), \eta)
        = \lambda^{\Obs}(\delta^{\Obs}(r', \sigma_2), \eta).
        \tag*{$(*)$}
      \end{equation}
      We verify that $\access(q')\,\sigma_2$ can be passed to
      $\ProcessCounterexample$:
      \begin{align*}
        \lambda^{\Obs}(\delta^{\Hyp}(q_0^\Hyp, \access(q')\,\sigma_2), \eta)
        &= \lambda^{\Obs}(\delta^{\Hyp}(q',\sigma_2), \eta)
        = \lambda^{\Obs}(q, \eta)
        \intertext{But on the other hand:}
        \lambda^{\Obs}(\delta^{\Obs}(q_0^\Obs, \access(q')\,\sigma_2), \eta)
        &= \lambda^{\Obs}(\delta^{\Obs}(q',\sigma_2), \eta)
        \overset{(*)}{=} \lambda^{\Obs}(\delta^{\Obs}(r',\sigma_2), \eta)
        \\
        &= \lambda^{\Obs}(r, \eta)
        \neq \lambda^{\Obs}(q, \eta)
      \end{align*}
      Hence, $\eta$ is a witness for $\delta^{\Hyp}(q_0^\Hyp,
      \access(q')\,\sigma_2) \apart \delta^{\Obs}(q_0^\Obs,
      \access(q')\,\sigma_2)$ and invoking $\ProcessCounterexample(\access(q')\, \sigma_2)$
      makes that $\Hyp$ is not a hypothesis for $\Obs$ afterwards.
      \qedhere
    \end{enumerate}
  \end{itemize}
\end{proofappendix}

\begin{proofappendix}{adscorrect}

  We first show that if there are at least two states to be distinguished, 
  the expected reward is positive:
  \begin{lemma}
    \label{adsexists}
    Suppose that $U\subseteq Q^{\Obs}$, $r, r'\in U$ with $r \apart r'$. Then $E(U) > 0$.
  \end{lemma}
  \begin{proof}
    Let $\sigma \in I^+$ be a witness for $r\apart r'$, that is $\sigma\vdash r\apart r'$.
    We prove $E(U) > 0$ by induction on the length
      of $\sigma$. Let $\sigma = i\;\sigma'$ for $i\in I$ and
      $\sigma'\in I^*$ and put $o = \lambda(r,i)\in O$, and $o' = \lambda(r',i)
      \in O$.
      \begin{itemize}
      \item If $o\neq o'$, then $\delta(r,i) \in U\xrightarrow{i/o}$ and
        $\delta(r',i) \not\in U\xrightarrow{i/o}$, and so
        \[
          |U\xrightarrow{i/o}| \ge 1
          \text{ and }
          (|U\xrightarrow{i}| - |U\xrightarrow{i/o}|) \ge 1
        \]
        Thus, the fraction in \eqref{defE} is greater than $0$ and consequently also $E(U) > 0$.

      \item If $o=o'$, then
        $\delta(r,i),\delta(r',i) \in U\xrightarrow{i/o}$
        and $\sigma' \vdash \delta(r,i)\apart\delta(r',i)$. Hence, we obtain
        $E(U\xrightarrow{i/o}) > 0$ by the induction hypothesis and so the
        fraction for $i$ and $o$ is greater than 0 and so $E(U) > 0$.
        \qed
      \end{itemize}
  \end{proof}

Next we show that if the expected reward is positive, there exist, for every maximal path in $\Ads(U)$, two states in $U$ that are apart.

  \begin{lemma}
    \label{adsproductive}
    If $E(U) > 0$, $\pi$ is a path from the root of $\Ads(U)$ to a leaf, and $\sigma\in I^*$ is the sequence of labels of the states occurring in $\pi$, then there are
    $r,r'\in U$ with $\lambda(r,\sigma) \neq \lambda(r',\sigma)$.
  \end{lemma}
  \begin{proof}
    We prove the claim by induction on the length (that is, the number of transitions) of $\pi$.
    
    If the length of $\pi$ is $0$ then none of the states in $U$ has an outgoing transition, which implies that $E(U)=0$, which means that the statement of the lemma holds.

    For the induction step, assume that the length of $\pi$ is greater than $0$.
    Then at least one state in $U$ has an outgoing transition.
    Let $i$ be the input that witnesses the maximum in $E(U)>0$.
    Then the root of $\Ads(U)$ has label $i$ and thus $\sigma = i \;\sigma'$, for some $\sigma'$.
    Suppose that path $\pi$ starts with an $o$-transition from $U$ to $\Ads(U\xrightarrow{i/o})$,
    where $\mathord{U\xrightarrow{i/o}} \neq \emptyset$. We consider two cases:
    \begin{itemize}
    \item If $|U\xrightarrow{i}| > |U\xrightarrow{i/o}|$, then there is some
      $o'\neq o$ with $\mathord{U\xrightarrow{i/o'}}\neq \emptyset$. This implies there exist states $r,r'\in U$ with $\lambda(r,i) =o$
      and $\lambda(r,i) =o'$
      and in particular $\lambda(r,\sigma) \neq \lambda(r',\sigma)$.
    \item If $|U\xrightarrow{i}| = |U\xrightarrow{i/o}|$ then
    $E(U\xrightarrow{i/o}) = E(U) > 0$. This means we can apply the induction hypothesis to obtain $p,p'\in U\xrightarrow{i/o}$
      with $\lambda(p,\sigma') \neq \lambda(p',\sigma')$.
      By definition of $U\xrightarrow{i/o}$, this yields us $r,r'\in U$ with
      $r\xrightarrow{i/o} p$ and
      $r'\xrightarrow{i/o} p'$
      and so $\lambda(r, \sigma) \neq\lambda(r, \sigma)$ for the composed
      $\sigma = i\;\sigma'$.
      \qed
    \end{itemize}
  \end{proof}
  We now come to the main proof of \autoref{adscorrect}, i.e.~that the updated
  output queries induce at least the same apartness pairs and so the norm
  $\N(\Obs)$ grows with each rule application:
  \begin{enumerate}
  \item It is clear that $\OutputQuery(\access(q)\;i\;\Ads(S))$ in the updated (R2')
    discovers at least the same apartness pairs as $\OutputQuery(\access(q)\;i)$
    in plain $\lsharp$.
  \item For $U := \{b\in S\mid \neg(b\apart q)\}$, we show that 
    querying $\access(q)\;\Ads(U)$ in the updated
    (R3') makes $b$ apart from at least one $r\in U$.
    Whenever the rule (R3') is applied, then $U$ contains two states that are apart, so $E(U) > 0$ by
    \autoref{adsexists}. Let $\sigma\in I^+$ be the sequence that is sent in
    total to the teacher, i.e.~$\sigma$ is the path of $\Ads(U)$ that is
    actually run.

    We distinguish two cases:
    \begin{itemize}
    \item If $\sigma$ does not reach to a leaf in the decision tree $\Ads(U)$,
      then this means that the adaptive distinguished sequence terminated
      earlier because of an unexpected output $o\in O$. Concretely, this means
      that $\lambda(q,\sigma) \neq \lambda(b,\sigma)$ for all $b\in U$.
    \item If $\sigma$ reaches a leaf in the decision tree $\Ads(U)$,
      then we obtain that there are $r,r'\in U$ with $\lambda(r,\sigma) \neq
      \lambda(r',\sigma)$, by \autoref{adsproductive}. After the output query,
      $\lambda(q,\sigma)\converges$, and so by weak co-transitivity (\autoref{la:
        weak co-transitivity}), either $q\apart r$ or $q\apart r'$ or both.
  \end{itemize}
  Hence, $\lsharp_\Ads$ lets the norm $\N(\Obs)$ grow with each rule
  application. By \autoref{thmBound}, $\lsharp_\Ads$ must have reached the
  correct hypothesis before $\N(\Obs)$ exceeds the bound, hence it is correct.
  (Note that \autoref{thmBound} is a general observation on observation trees
  and does not involve the algorithm at all).
  \end{enumerate}
\end{proofappendix}

\begin{proofappendix}{thmQueryComplexity}
  Strategic $\lsharp_\Ads$ makes the same amount of output queries and
  equivalence queries as strategic $\lsharp$, so it is sufficient to discuss
  strategic $\lsharp$.

	In the strategic $\lsharp$, every (non-terminating) application of rule (R4) leads to
	an isolated state in the frontier, i.e.~increases the basis by one state before
	another equivalence query can be asked.
	Since the basis may contain at most $n$ elements, this means that there are at most
	$n-1$ applications of rule (R4).
	Processing the counterexamples generated by the resulting consistency checks and equivalence queries of rule (R4) will require
	$\bigO(n \log m)$ output queries.
	By Theorem~\ref{thmProgress} and Theorem~\ref{thmBound} there are at most $\bigO(k n^2)$ rule applications during a run of $\lsharp$.
	Since applications of rule (R1) require no output queries, and each application of rule (R2) and (R3) requires exactly one output query, this means that applications of rules (R1), (R2) and (R3) will require $\bigO(k n^2)$ output queries.
	Altogether, $\lsharp$ will require $\bigO(kn^2+n \log m)$ output queries.
\end{proofappendix}

\begin{proofappendix}{thmSymbolComplexity}
	During counterexample processing, we may create a witness of length at most $m$ between a state $q$ in the frontier and a state $r$ in the basis. When we
	subsequently move state $q$ to the basis, we have created a pair of states in the basis with a witness of length $m$.
	In fact, at any point during a run of $\lsharp$, the length of a minimal witness that distinguishes a state from $F \cup S$ from a state of $S$ will be at most $m$.
	Since $n \in \bigO(m)$, this implies that the number of symbols in any output query will be $\bigO(m)$.
	Since, by Theorem~\ref{thmQueryComplexity}, there are $\bigO(kn^2+n \log m)$ output queries, the result follows.
\end{proofappendix}

  \newpage
  \section{Complete benchmarking results}%
  \label{completeBenchmarks}%
  \newcommand{\totalpagecount}{5}%
  In Tables \ref{table_full_1} to \ref{table_full_\totalpagecount},
  we list the number of queries for every model and every
  learning algorithm. See \autoref{sec:bench} for and description of the
  benchmark setup.
  \foreach \pagenr in {1,...,\totalpagecount} {%
    \begin{table}[H]
      \caption{Benchmark results (Part \pagenr\ of \totalpagecount)}
      \label{table_full_\pagenr}
      \scriptsize{}\centering
      \input{table_full_\pagenr}
    \end{table}
  }
}{}

\end{document}